\documentclass[orivec,a4paper,10pt]{llncs}
\usepackage{cite}
\usepackage{mathtools}
\usepackage{xspace}
\usepackage{amssymb}
\usepackage{xcolor}
\usepackage{graphicx}
\usepackage{thmtools}
\usepackage{thm-restate}
\usepackage{hyperref}
\newcommand{\naturals}{\mathbb{N}}

\newcommand{\hide}[1]{}
\newcommand{\setnocond}[1]{\{#1\}}
\newcommand{\setcond}[2]{\{\, #1 \mid #2 \,\}}

\newcommand{\wordletter}[2]{#1{[#2]}}
\newcommand{\subword}[3]{#1 {[#2..#3]}}

\newcommand{\pathto}[2]{{\xrightarrow[]{{#1}}}}
\newcommand{\fpathto}[2]{{\xRightarrow[]{#1}}}

\newcommand{\vertex}[2]{\langle {#1}, {#2} \rangle}

\newcommand{\alphabet}{\Sigma}
\newcommand{\emptyword}{\epsilon}
\newcommand{\finwords}{\alphabet^*}
\newcommand{\infwords}{\alphabet^\omega}
\newcommand{\poswords}{\alphabet^+}
\newcommand{\langsymb}[0]{\mathcal{L}}
\newcommand{\lang}[1]{\langsymb(#1)}
\newcommand{\finlang}[1]{\langsymb_{*}(#1)}
\newcommand{\inflang}[1]{\langsymb(#1)}

\newcommand{\upword}[1]{\text{UP}(#1)}

\newcommand{\states}{Q}

\newcommand{\trans}{\delta}

\newcommand{\inits}{I}

\newcommand{\acc}{F}

\newcommand{\run}{\rho}

\newcommand{\dirac}[1][\acc]{\mathbf{1}_{#1}}

\newcommand{\A}{\mathcal{A}}
\newcommand{\B}{\mathcal{B}}
\newcommand{\C}{\mathcal{C}}
\newcommand{\D}{\mathcal{D}}

\newcommand{\F}{\mathcal{F}}
\newcommand{\R}{\mathtt{T}}
\newcommand{\M}{\mathcal{M}}
\newcommand{\N}{\mathcal{N}}

\newcommand{\bigO}{\mathcal{O}}

\newcommand{\buchi}{B\"uchi\xspace}
\newcommand{\rbc}{\textsf{RBC}\xspace}

\newcommand{\size}[1]{|{#1}|}
\newcommand{\authorComment}[3]{{\color{#1}\textbf{[\!\![\!\![\marginpar{\centering{\color{#1}\textbf{#2}}}~#2: #3 ---#2.~]\!\!]\!\!]}}}

\newcommand{\at}[1]{\authorComment{teal}{AT}{#1}}

\newcommand{\td}[1]{\textcolor{blue}{\ifmmode \text{[#1]}\else [#1] \fi}}

\newcommand{\infstates}[1]{\textit{Inf}({#1})}

\newcommand{\canoEq}{\backsim}
\newcommand{\icanoEq}{\backsim^{i}}
\newcommand{\proEq}{\approx}

\newcommand{\proclass}[1]{[#1]_{{\proEq}_{u}}}

\newcommand{\oproclass}[1]{[#1]_{{\proEq}^{o}_{u}}}
\newcommand{\class}[1]{[#1]_{\canoEq}}
\newcommand{\iclass}[1]{[#1]_{\icanoEq}}
\newcommand{\oclass}[1]{[#1]_{\canoEq^{o}}}
\newcommand{\quotient}{\finwords/_{\canoEq}}
\newcommand{\iquotient}{\finwords/_{\icanoEq}}
\newcommand{\oquotient}{\finwords/_{\canoEq^{o}}}
\newcommand{\classpreq}{\trianglelefteq}
\newcommand{\classpre}{\triangleleft}

\allowdisplaybreaks

\begin{document}

\title{Congruence Relations for \buchi Automata}

\author{Yong Li\inst{1}, Yih-Kuen Tsay\inst{2}, Andrea Turrini\inst{1}, Moshe Y. Vardi\inst{3}, Lijun Zhang\inst{1}}
\institute{
State Key Laboratory of Computer Science, \\
Institute of Software, Chinese Academy of Sciences
\and
National Taiwan University
\and
Rice University
}
\maketitle
\thispagestyle{plain}

\begin{abstract}
We revisit congruence relations for \buchi automata, which play a central role in the automata-based formal verification. The size of the classical congruence relation is in $3^{\bigO(n^{2})}$, where $n$ is the number of states of a given \buchi automaton $\A$. We present improved congruence relations that can be exponentially coarser than the classical one. We further give asymptotically \emph{optimal} congruence relations of size $2^{\bigO(n \log n)}$. Based on these optimal congruence relations, we obtain an \emph{optimal} translation from \buchi automata to a family of deterministic finite automata (FDFW) that accepts the complementary language. To the best of our knowledge, our construction is the \emph{first direct} and optimal translation from \buchi automata to FDFWs.
\end{abstract}

\vspace{-0.0mm}
\section{Introduction}\label{sec:intro}
\vspace{-0.0mm}
Congruence relations for nondeterministic \buchi automata on words (NBWs)~\cite{Buc62} are fundamental for \buchi complementation, a key operation used in the formal verification framework based on automata theory~\cite{DBLP:conf/lics/VardiW86}.
To formally verify whether the behavior of a system $A$ satisfies the given specification $B$, one usually reduces this problem to a language-containment problem between the NBWs $A$ and $B$, which is reduced to the intersection of $A$ and the complement of $B$.
The first complementation construction for \buchi automata, proposed by \buchi~\cite{Buc62} and widely known as Ramsey-based \buchi complementation (\rbc), relies on a congruence relation with a $2^{2^{\bigO(n)}}$ blow-up, where $n$ is the number of states of the input automaton.
One can associate each equivalence class of the congruence relation with a state of the complementary automaton, as happens by the Myhill-Nerode theorem for regular languages~\cite{DBLP:books/el/leeuwen90/Thomas90}.
The blow-up of the congruence relation of \rbc was later reduced by Sistla \emph{et al.}~\cite{sistla1987complementation} to $3^{\bigO(n^{2})}$, without providing an explicit formal notion of congruence relation, later formalized by~\cite{DBLP:books/el/leeuwen90/Thomas90}.

Notably, current practical approaches to the containment checking for NBWs are all based on the classical congruence relation given in~\cite{sistla1987complementation,DBLP:books/el/leeuwen90/Thomas90}, even though it has a larger blow-up ($3^{\bigO(n^{2})}$ vs. $2^{\bigO(n \log n)}$) than other optimal complementation constructions, such as \emph{rank-based} complementation~\cite{KV01d}.
In fact, RABIT, also based on the classical congruence relation of~\cite{sistla1987complementation,DBLP:books/el/leeuwen90/Thomas90}, is the state-of-the-art tool for checking language-containment between NBWs, and has integrated various state-space pruning techniques for \rbc proposed in~\cite{DBLP:conf/cav/AbdullaCCHHMV10, DBLP:conf/concur/AbdullaCCHHMV11, DBLP:journals/lmcs/ClementeM19}.

Recently, \emph{families of deterministic finite automata on words} (FDFWs)~\cite{AngluinBF18} have been proposed for representing $\omega$-regular languages, as an alternative to NBWs.
If we model system and specifications as FDFWs, then the formal verification problem can be reduced to a containment problem between two FDFWs, which is doable in polynomial time~\cite{AngluinBF18}, in contrast to PSPACE-complete for NBWs~\cite{DBLP:conf/cav/KupfermanV96a}.
It has been shown that FDFWs can be induced from a congruence relation defined over a given $\omega$-regular language, where each state of the FDFW corresponds to an equivalence class of the congruence relation~\cite{AngluinF16}.

In this work we show that \rbc and FDFWs have an intimate connection and the underlying concept that connects them is the congruence relation for NBWs.
This connection gives us the possibility to further tighten the congruence relations for both \rbc and FDFWs (see Sects.~\ref{sec:improved-ramsey-for-fdfws} and~\ref{sec:optimal-rc}).
In fact, the state-space pruning techniques developed in~\cite{DBLP:conf/cav/AbdullaCCHHMV10, DBLP:conf/concur/AbdullaCCHHMV11} for \rbc~\cite{sistla1987complementation} are inherently heuristics for identifying subsumption and simulation relations between congruence relations of \rbc. Therefore, in order to further theoretically or empirically improve model-checking algorithms based on \rbc or FDFWs, it is important to understand the congruence relations for both FDFWs and \rbc and, hopefully, make their congruence relations smaller.

\emph{Contribution.}
We focus here on an in-depth study of the congruence relation for NBWs and its connection to FDFWs.
First, we show how to improve the classical congruence relation $\canoEq$ with a blow-up of $3^{\bigO(n^{2})}$ defined by the classical \rbc to congruence relations that can be exponentially tighter (Theorem~\ref{thm:comparisonCanoEq}), while
the improved congruence relations can never be larger than the classical congruence relation $\canoEq$ (Theorem~\ref{thm:compactness}).
Notably, the improved congruence relations only have a blow-up of $\bigO(n^{2})$ when dealing with deterministic \buchi automata (Theorem~\ref{thm:size-dbw-improved-fdfw}).
Second, we further propose congruence relations for NBWs with a blow-up of only $2^{\bigO( n \log n)}$ (Lemma~\ref{lem:upper-bound-opt-proEq}), which is then proved to be optimal (Theorem~\ref{thm:lower-bound-rc}).
Finally, we show that our congruence relations define an FDFW recognizing $\infwords \setminus \lang{\A}$ from an NBW $\A$.
In particular, if $\A$ has $n$ states, then our optimal congruence relations yield an FDFW $\F$ with an optimal complexity $2^{\bigO(n \log n)}$.
Thus, to the best of our knowledge, we present the \emph{first} \emph{direct} translation from an NBW to an FDFW with \emph{optimal} complexity.

We defer all missing proofs of the paper to Appendix~\ref{app:proofs}.
\vspace{-0.0mm}
\section{Preliminaries}\label{sec:preliminaries}
\vspace{-0.0mm}

Fix an \emph{alphabet} $\alphabet$.
A \emph{word} is a finite or infinite sequence of letters in $\alphabet$.
Let $\finwords$ and $\infwords$ denote the set of all finite and infinite words (or $\omega$-words), respectively.
A \emph{finitary language} is a subset of $\finwords$;
an \emph{$\omega$-language} is a subset of $\infwords$.
Let $L$ be a finitary language (resp.,~$\omega$-language); the complementary language of $L$ is $\finwords \setminus L$ (resp., $\infwords \setminus L$).
Let $\run$ be a sequence;
we denote by $\wordletter{\run}{i}$ the $i$-th element of $\run$ and by $\subword{\run}{i}{k}$ the subsequence of $\run$ starting at the $i$-th element and ending at the $k$-th element, included, when $i \leq k$, and the empty sequence $\emptyword$ when $i > k$.
Given a finite word $u$ and a word $w$, we denote by $u \cdot w$ ($uw$, for short) the concatenation of $u$ and $w$.
Given a finitary language $L_{1}$ and a finitary/$\omega$-language $L_{2}$, the concatenation $L_{1} \cdot L_{2} $ ($L_{1} L_{2}$, for short) of $L_{1}$ and $L_{2}$ is the set $L_{1} \cdot L_{2} = \setcond{uw}{u\in L_{1}, w \in L_{2}}$ and $L^{\omega}_{1}$ the infinite concatenation of $L_{1}$.

\paragraph{NBWs.}
A (nondeterministic) automaton is a tuple $\A = (\states, \inits, \trans, \acc)$, where $\states$ is a finite set of states, $\inits \subseteq \states$ is a set of initial states, $\trans \colon \states \times \alphabet \rightarrow 2^{\states}$ is a transition function, and $\acc \subseteq \states$ is a set of accepting states.
We extend $\trans$ to sets of states, by letting $\trans(S, a) = \bigcup_{q \in S} \trans(q, a)$.
We also extend $\trans$ to words, by letting $\trans(q, \emptyword) = \setnocond{q}$ and $\trans(q, a_{1} a_{2} \cdots a_{k}) = \trans(\trans(q, a_{1}),\cdots, a_{k})$, where we have $k \geq 1$ and $a_{i} \in \alphabet$ for $i \in \setnocond{1, \cdots,k}$.
An automaton on finite words is called a \emph{nondeterministic automaton on finite words} (NFW), while an automaton on $\omega$-words is called a \emph{nondeterministic \buchi automaton on infinite words} (NBW).
An NFW $\A$ is said to be a \emph{deterministic} finite-word automaton (DFW) if $\size{\inits} = 1$ and for each $q \in \states$ and $a \in \alphabet$, $\size{\trans(q, a)} \leq 1$.
DBWs are defined similarly.

A \emph{run} of an NFW/NBW $\A$ on a finite word $u$ of length $n \geq 0$ is a sequence of states $\run = q_{0} q_{1} \cdots q_{n} \in \states^{+}$, such that for every $0 < i \leq n$, $q_{i} \in \trans(q_{i-1}, \wordletter{u}{i})$.
We write $q_{0} \pathto{u}{\trans} q_{n}$ if there is a run from $q_{0}$ to $q_{n}$ over $u$ and by $q_{0} \fpathto{u}{\trans} q_{n}$ if such run also visits an accepting state.
Obviously, we have that $q \pathto{\emptyword}{\trans} q$ for all $q \in \states$ and $q \fpathto{\emptyword}{\trans} q$ for all $q \in \acc$.
A finite word $u \in \finwords$ is \emph{accepted} by an NFW $\A$ if there is a run $q_{0} \cdots q_{n}$ over $u$ such that $q_{0} \in \inits$ and $q_{n} \in \acc$.
Similarly, an \emph{$\omega$-run} of $\A$ on an $\omega$-word $w$ is an infinite sequence of states $\run = q_{0} q_{1}\cdots$ such that $q_{0} \in \inits$ and for every $i > 0$, $q_{i} \in \trans(q_{i-1}, w[i])$.
Let $\infstates{\run}$ be the set of states that occur infinitely often in the run $\run$.
An $\omega$-word $w \in \infwords$ is \emph{accepted} by an NBW $\A$ if there exists an $\omega$-run $\run$ of $\A$ over $w$ such that $\infstates{\run} \cap \acc \neq \emptyset$.
The \emph{finitary language} recognized by an NFW $\A$, denoted by $\finlang{\A}$, is defined as the set of finite words accepted by it. Similarly, we denote by $\inflang{\A}$ the \emph{$\omega$-language} recognized by an NBW $\A$, i.e., the set of $\omega$-words accepted by $\A$.
The complementary automaton of an NBW $\A$, denoted as $\A^{c}$, accepts the complementary language of $\inflang{\A}$, i.e., $\inflang{\A^{c}} = \infwords \setminus \inflang{\A}$.

\paragraph{Congruence Relations.}
A \emph{right congruence} (RC) relation is an equivalence relation $\canoEq$ over $\finwords$ such that $x \canoEq y$ implies $xv \canoEq yv$ for all $v \in \finwords$.
A \emph{congruence relation} is an equivalence relation $\canoEq$ over $\finwords$ such that $x \canoEq y$ implies $uxv \canoEq uyv$ for every $x, y, u, v \in \finwords$.
We denote by $\size{\canoEq}$ the index of $\canoEq$, i.e., the number of equivalence classes of $\canoEq$.
A \emph{finite congruence relation} is a congruence relation with a finite index.
We use $\quotient$ to denote the set of equivalence classes of $\canoEq$.
Given $x \in \finwords$, we denote by $\class{x}$ the equivalence class of $\canoEq$ that $x$ belongs to.

For a given right congruence $\canoEq$ of a regular language $L$, it is well-known that Myhill-Nerode theorem~\cite{Myhill57,Nerode58} defines a unique minimal DFW $D$ of $L$, in which each state of $D$ corresponds to an equivalence class defined by $\canoEq$ over $\finwords$.
Therefore, we can construct a DFW $\D[\canoEq]$ from $\canoEq$ in a standard way.

\begin{definition}[\hspace*{-1.5mm}\cite{Myhill57,Nerode58}]
\label{def:induced-dfw}
Let $\canoEq$ be a right congruence of finite index.
The DFW $\D[\canoEq]$ without accepting states induced by $\canoEq$ is a tuple $(S, s_{0}, \trans_{\D}, \emptyset)$ where
$S = \quotient$, $s_{0} = \class{\emptyword}$, and for each $u \in \finwords$ and $a \in \alphabet$, $\trans_{\D}(\class{u}, a) = \class{ua}$.
\end{definition}

The DFW $\D[\canoEq]$, parametric on $\canoEq$, indicates that it is induced by the right congruence relation $\canoEq$.
We may just write $\D$ if $\canoEq$ is clear from the context.

\paragraph{UP-words.}
The $\omega$-regular languages accepted by NBWs can also be recognized by FDFWs  by means of their \emph{ultimately periodic words} (UP-words)~\cite{AngluinBF18}.
A UP-word $w$ is an $\omega$-word of the form $uv^{\omega}$, where $u \in \finwords$ and $v \in \poswords$.
Thus $w = uv^{\omega}$ can be represented as a pair of finite words $(u, v)$, called a \emph{decomposition} of $w$.
A UP-word can have multiple decompositions: for instance $(u, v)$, $(uv, v)$, and $(u, vv)$ are all decompositions of $uv^{\omega}$.
For an $\omega$-language $L$, let $\upword{L} = \setcond{uv^{\omega} \in L }{u \in \finwords, v \in \poswords }$ denote the set of all UP-words in $L$.
The set of UP-words of an $\omega$-regular language $L$ can be seen as the
fingerprint of $L$, as stated below.
\begin{theorem}[\hspace*{-1.4mm}\cite{Buc62}]
\label{thm:upword-omega-regular-language}
(1)
Every non-empty $\omega$-regular language $L$ contains at least one UP-word.
(2)		
Let $L$, $L'$ be two $\omega$-regular languages.
Then $L = L'$ if and only if $\upword{L} = \upword{L'}$.
\end{theorem}

\paragraph{FDFWs.}
Based on Theorem~\ref{thm:upword-omega-regular-language}, Angluin \emph{et al.} introduced in~\cite{AngluinBF18} the notion of FDFWs as another type of automata to recognize $\omega$-regular languages.
\begin{definition}[FDFWs~\cite{AngluinBF18}]
\label{def:fdfws}
An FDFW is a pair $\F = (\M, \setnocond{\N_{q}})$ consisting of a leading DFW $\M$ and of a progress DFW $\N_{q}$ for each state $q$ in $\M$.
\end{definition}
Intuitively, the leading DFW $\M$ of $\F = (\M, \setnocond{\N_{q}})$ for an $\omega$-regular language $L$ consumes the finite prefix $u$ of a UP-word $uv^{\omega} \in \upword{L}$, reaching some state $q$, and for each state $q$ of $\M$, the progress DFW $\N_{q}$ accepts the period $v$ of $uv^{\omega}$.
An example of FDFW $\F$ is depicted in Fig.~\ref{fig:fdfw-example} where the leading DFW $\M$ has only one state $s$ and the progress DFW associated with $s$ is $\N_{s}$.
Note that the leading DFW $\M$ of every FDFW does not make use of accepting states.
\begin{figure}[t]
\begin{center}
\vspace*{-2mm}
\includegraphics[scale=0.9]{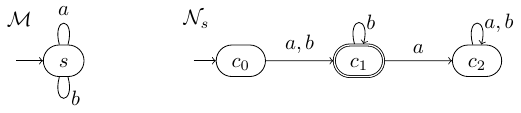}
\end{center}
\vspace*{-3mm}
\caption{An example of FDFW $\F = (\M, \setnocond{\N_{s}})$ which is not saturated.}
\label{fig:fdfw-example}
\end{figure}

Let $\D$ be a DFW with initial state $q_{0}$ and transition function $\trans$.
Given a word $u \in \finwords$, we often use $\D(u)$ as a shorthand for $\trans(q_{0}, u)$.
Each FDFW $\F$ characterizes a set of UP-words $\upword{\F}$ by following the acceptance condition.
\begin{definition}[FDFW Acceptance]
\label{def:acc-fdfa}
Let $\F = (\M, \setnocond{\N_{q}})$ be an FDFW and $w$ be a UP-word.
A decomposition $(u, v)$ of $w$ is \emph{normalized} with respect to $\F$ if $\M(u) = \M(uv)$.\footnote{We use the normalized decomposition of UP-words defined in~\cite{li2017novel}, which is different from the one given in~\cite{AngluinBF18}. Ours is a definition for a UP-word, while their definition is applied to a decomposition. However, this difference does not affect the definition of a saturated FDFW to be given later.}
A decomposition $(u, v)$ is accepted by $\F$ if $(u, v)$ is normalized and we have $v \in \finlang{\N_{q}}$  where $q = \M(u)$.
$w$ is accepted by $\F$ if there exists a decomposition $(u, v)$ of $w$ accepted by $\F$.
\end{definition}
Note that the normalized decomposition $(u, v)$ is defined with respect to $\F$. We usually omit $\F$ and just say $(u, v)$ is normalized when $\F$ is clear from the context.
Consider again the FDFW $\F$ in Fig.~\ref{fig:fdfw-example}:
$(aba)^{\omega}$ is not accepted since no decomposition of $(aba)^{\omega}$ is accepted
while $(ab)^{\omega}$ is accepted since the decomposition $(ab, ab)$ of $(ab)^{\omega}$ is such that $\M(ab \cdot ab) = \M(ab) = s$ and $ab \in \finlang{\N_{s}}$.

One can observe that the normalized decomposition $(ab, abab)$ of $(ab)^{\omega}$ is not accepted by $\F$, despite that $(ab, ab)$ is accepted by $\F$. In the following, we define a class of FDFWs that \emph{saturates} each accepting normalized decomposition $(ab, (ab)^{k})$ of $(ab)^{\omega}$ (where $k \geq 1$) if $(ab, ab)$ is accepted, which is important for FDFWs to recognize $\omega$-regular languages~\cite{li2017novel,AngluinBF18}.

\begin{definition}[Saturation of FDFWs~\cite{AngluinBF18}]
\label{def:saturated-fdfws}
Let $\F = (\M, \setnocond{\N_{q}})$ be an FDFW and $w$ be a UP-word in $\upword{\F}$.
We say $\F$ is
\emph{saturated} if for all normalized decompositions $(u, v)$ and $(u', v')$ of $w$, either both $(u, v)$ and $(u', v')$ are accepted by $\F$ or both are not.
\end{definition}

Intuitively, for a saturated FDFW $\F$, a UP-word $w$ is accepted by $\F$ if and only if all normalized decompositions $(u, v)$ of $w$ are accepted by $\F$.
From a saturated FDFW $\F$, one can construct an equivalent NBW $\A$ that recognizes $\upword{\F}$ in polynomial time.
\begin{lemma}[Polynomial Translation from FDFWs to NBWs~\cite{li2017novel,AngluinBF18}]
\label{lem:fdfw-to-nbw}
Let $\F = (\M, \setnocond{\N_{q}})$ be a saturated FDFW with $n$ states.
Then, one can construct an NBW $\A$ with $\bigO(n^3)$ states such that $\upword{\F} = \upword{\lang{\A}}$.
\end{lemma}
Note that an FDFW that is \emph{not} saturated does not necessarily recognize an $\omega$-regular language (cf.~\cite{li2017novel}), let alone permit an equivalent translation to NBWs.

In the remainder of the paper, we fix an NBW $\A = (\states, \inits, \trans, \acc)$, unless explicitly stated otherwise, where $\A$ has $n$ states, i.e., $n = \size{\states }$.
We call a state in an FDFW a \emph{macrostate} to distinguish it from states of $\A$.

\section{Improved Congruence Relations for NBWs}
\label{sec:ramsey-for-fdfws}

In this section we present congruence relations that can be used to generate NBWs accepting the language of a given NBW $\A$ or its complement.
We first review in Sect.~\ref{ssec:classical-cr} the classical congruence relations defined in~\cite{sistla1987complementation,DBLP:books/el/leeuwen90/Thomas90} and then give improved congruence relations in Sect.~\ref{sec:improved-ramsey-for-fdfws}.

\subsection{Classical Congruence Relations}
\label{ssec:classical-cr}

As mentioned in the introduction, the index of the congruence relation of \rbc proposed by \buchi~\cite{Buc62} is doubly exponential in the size of $\A$.
Sistla, Vardi, and Wolper~\cite{sistla1987complementation} showed how to improve \rbc with a subset construction that was later presented by Thomas~\cite{DBLP:books/el/leeuwen90/Thomas90} as the following congruence relation $\canoEq_{\A}$.

\begin{definition}[\hspace*{-1.5mm}\cite{sistla1987complementation,DBLP:books/el/leeuwen90/Thomas90}]
\label{def:congruence-rel}
In the \rbc construction,
for all $u_{1}, u_{2} \in \finwords$, $u_{1} \canoEq_{\A} u_{2}$ if and only if for all $q, r \in \states$,
(1) $q \pathto{u_{1}}{\trans} r$ iff $q \pathto{u_{2}}{\trans} r$;
and (2) $q \fpathto{u_{1}}{\trans} r$ iff $q \fpathto{u_{2}}{\trans} r$.
\end{definition}

It is easy to verify that $\canoEq_{\A}$ is a (right-)congruence relation:
given two finite words $u_{1}$ and $u_{2}$ such that $u_{1} \canoEq_{\A} u_{2}$, we have that $xu_{1}y \canoEq_{\A} xu_{2}y$ holds for all $x,y \in \finwords$.
Moreover, we have that $\canoEq_{\A}$ is of finite index, as stated by the next lemma.
To simplify the notation, we just write $\canoEq$ instead of $\canoEq_{\A}$ as $\A$ is fixed.

\begin{restatable}[\hspace*{-1.5mm}\cite{sistla1987complementation,DBLP:books/el/leeuwen90/Thomas90}]{lemma}{sizeOfcanoEq}
\label{lem:size-of-canoEq-old}
Let $\canoEq$ be as given in Definition~\ref{def:congruence-rel}. Then $\size{\canoEq} \leq 3^{n^{2}}$.
\end{restatable}

Since the congruence relation $\canoEq$ is defined by reachability between states, the result follows from the fact that we can map each of the $n^2$ pairs of states $(q, r)$ to either both $q \fpathto{u}{} r$ and $q \pathto{u}{} r$, or just $q \pathto{u}{} r$ or none of them.
Thus we have $\size{\canoEq} = \size{\quotient} \leq 3^{n^{2}}$.
We can also establish a lower bound for $\canoEq$, by means of a family of DBWs inspired by the proof of~\cite[Theorem 2]{AngluinF16}.

\begin{restatable}[]{lemma}{sizeOfCongr}
\label{lem:lower-bound-congrel}
There is a family of DBWs $\setnocond{\C_{n}}_{n \in \naturals}$ such that  each DBW $\C_{n}$ has $n+2$ states and the corresponding $\canoEq_{\C_{n}}$ is such that $\size{\canoEq_{\C_{n}}} \geq n!$.
\end{restatable}

An important property we want to have is that the congruence relation $\canoEq$ captures correctly the language of the NBW it correspond to.
This means that $\canoEq$ must not relate words in $\lang{\A}$ with those in $\infwords \setminus \lang{\A}$, that is, for each $\class{u}\in \quotient$ and $\class{v} \in \poswords/_{\canoEq}$, either $\class{u}\class{v}^{\omega} \subseteq \lang{\A}$ or $\class{u}\class{v}^{\omega} \subseteq \infwords \setminus \lang{\A}$.
Moreover, $\canoEq$ should cover the whole $\infwords$, that is, it saturates $\lang{\A}$, $\infwords \setminus \lang{\A}$ and $\infwords$.
This is formalized by the following \emph{saturation lemma} of the congruence relation $\canoEq$, which is a known result from~\cite{sistla1987complementation} that we adapt to our notations.

According to~\cite{sistla1987complementation,DBLP:books/el/leeuwen90/Thomas90}, given two classes $\class{u} \in \quotient, \class{v} \in \poswords/_{\canoEq}$, the $\omega$-language $\class{u}\class{v}^{\omega}$ is called \emph{proper} if $\class{u} \class{v} \subseteq \class{u}$ and $\class{v} \class{v} \subseteq \class{v}$.

\begin{lemma}[Saturation Lemma~\cite{sistla1987complementation,DBLP:books/el/leeuwen90/Thomas90}]
\label{lem:saturation-congruence}
\begin{enumerate}
\item
    For $\class{u} \in \quotient, \class{v} \in \poswords/_{\canoEq}$, if $\class{u}\class{v}^{\omega}$ is proper, then either $\class{u}\class{v}^{\omega} \cap \inflang{\A}= \emptyset$ or $\class{u}\class{v}^{\omega}\cap \subseteq \inflang{\A}$.
\item
    $\infwords = \bigcup \setcond{\class{u}\class{v}^{\omega}}{\class{u} \in \quotient, \class{v} \in \poswords/_{\canoEq}, \text{$\class{u}\class{v}^{\omega}$ is proper}}$.

\item
    $\infwords \setminus \lang{\A} = \bigcup \setcond{\class{u}\class{v}^{\omega}}{\class{u} \in \quotient, \class{v} \in \poswords/_{\canoEq}, \class{u}\class{v}^{\omega} \cap \lang{\A} = \emptyset, \text{$\class{u}\class{v}^{\omega}$ is proper}}$.
\end{enumerate}
\end{lemma}
Thus, it suffices to just consider proper languages to get the languages $\infwords$ (cf. Item~(2) of Lemma~\ref{lem:saturation-congruence}) and $\infwords \setminus \lang{\A}$ (cf. Item~(3) of Lemma~\ref{lem:saturation-congruence}).
This means that the congruence relation $\canoEq$ allows us to obtain $\lang{\A}$ (resp., $\infwords \setminus \lang{\A}$) by identifying the exact set of proper languages that are inside $\lang{\A}$ (resp., outside $\lang{\A}$).
In the remainder of the paper, we show that we can obtain similar saturation lemmas (cf.~Lemmas~\ref{lem:saturation-impoved-fdfws} and~\ref{lem:saturation-opt-fdfws}) for the congruence relations we are going to propose to obtain $\lang{\A}$ or the complementary language $\infwords \setminus \lang{\A}$.

\subsection{Improved Congruence Relations for NBWs}
\label{sec:improved-ramsey-for-fdfws}

In this section, we introduce relations $\icanoEq$ and $\proEq_{u}$, $u \in \finwords$, that can never have larger index than the classical congruence relation $\canoEq$ (cf.~Lemma~\ref{lem:size-cano-eq}) while possibly being exponentially smaller than $\canoEq$ (cf.~Theorem~\ref{thm:comparisonCanoEq}).
When restricted to DBWs, we reduce the worst-case blow-up from $\Omega(n!)$ (cf.~Lemma~\ref{lem:lower-bound-congrel}) to $\bigO(n^{2})$ (cf.~Theorem~\ref{thm:size-dbw-improved-fdfw}).
Still, they capture correctly $\lang{\A}$ and $\infwords \setminus \lang{\A}$ (cf.~Lemma~\ref{lem:saturation-impoved-fdfws}).

We improve the classical congruence relation $\canoEq$ given in  Sect.~\ref{ssec:classical-cr} based on the following key observations:
(1) we can use different congruence relations to process the finite prefix $u$ and the periodic word $v$ of a UP-word $uv^{\omega}$, separately, in a manner similar to FDFWs;
(2) the assumption $\class{v} \class{v} \subseteq \class{v}$ in proper languages is not necessary, according to~\cite{sistla1987complementation};
(3) inspired by~\cite{DBLP:conf/fossacs/BreuersLO12}, we can consider only reachable states in $\A$, which allows us to use just \emph{right congruences} instead of congruences as $\canoEq$.
We defer the comparison of our work and~\cite{DBLP:conf/fossacs/BreuersLO12} to Remark~\ref{rmk:comparison}.

Instead of considering every pair of states $(q, r)$ of $\A$ to define the congruence relation $\canoEq$ (cf.~Definition~\ref{def:congruence-rel}), we process the finite prefixes $u$ by a simple subset construction over the states of $\A$, obtaining the following relation $\icanoEq$ that is obviously a right congruence.

\begin{definition}[RC $\icanoEq$]
\label{def:right-congruence}
For $u_{1}, u_{2} \in \finwords$, $u_{1} \icanoEq u_{2}$ iff $\trans(\inits, u_{1}) = \trans(\inits, u_{2})$.
\end{definition}

As one can expect, by relaxing the conditions on the relation $\icanoEq$, we reduce how large its index can be, from $3^{n^{2}}$ (cf.~Lemma~\ref{lem:lower-bound-congrel}) to $2^{n}$, showing also that $\icanoEq$ is a right congruence of finite index.

\begin{restatable}{lemma}{sizeOfCanoEq}
\label{lem:size-cano-eq}
Let $\icanoEq$ be the right congruence in Definition~\ref{def:right-congruence}.
Then $\size{\icanoEq} \leq 2^{n}$.
\end{restatable}

Differently from $\canoEq$ (see, e.g., Lemma~\ref{lem:saturation-congruence}), we will use $\icanoEq$ only to process the finite prefix $u$ of a UP-word $uv^{\omega}$;
to process the period $v$, we now introduce the right congruence $\proEq_{u}$, by considering only states reachable from $\trans(\inits, u)$.

\begin{definition}[RC $\proEq_{u}$]
\label{def:pro-right-congruence}
For $u, v_{1}, v_{2} \in \finwords$, we say $v_{1} \proEq_{u} v_{2}$ iff for all states $q \in \trans(\inits, u) $ and $r \in \states$ of $\A$, (1) $q \pathto{v_{1}}{\trans} r$ iff $q \pathto{v_{2}}{\trans} r$ and (2) $q \fpathto{v_{1}}{\trans} r$ iff $q \fpathto{v_{2}}{\trans} r$.
\end{definition}

Compared to Definition~\ref{def:congruence-rel}, we only take into account the states that can be reached from $\trans(I, u)$, as opposed to the whole set $\states$.
In this way we obtain a right congruence relation that is coarser than $\canoEq$ for the periodic finite words.

\begin{restatable}{theorem}{compactnessOfImprovedRc}
\label{thm:compactness}
Let $\canoEq$ be the congruence relation defined in Definition~\ref{def:congruence-rel}.
For each $u, v_{1}, v_{2} \in \finwords$, we have that $v_{1} \canoEq v_{2}$ implies that $v_{1} \proEq_{u} v_{2}$.

Similarly, $u_{1} \canoEq u_{2}$ implies that $u_{1} \icanoEq u_{2}$ for all $u_{1}, u_{2} \in \finwords$.
\end{restatable}

Although the right congruence relation $\proEq_{u}$ is coarser than its predecessor $\canoEq$, it has the same upper bound for its index (cf.~Lemma~\ref{lem:size-of-canoEq-old}).

\begin{restatable}{lemma}{sizeOfproEq}
\label{lem:size-pro-eq}
Given $u \in \finwords$, let $\proEq_{u}$ be as defined in Definition~\ref{def:pro-right-congruence}.
Then $\size{\proEq_{u}} \leq 3^{n^{2}}$.
\end{restatable}

Despite this common upper bound, $\size{\proEq_{u}}$ can be exponentially smaller than $\size{\canoEq}$, as witnessed by the family of NBWs $\setnocond{\B_{n}}_{n \in \naturals}$ depicted in Fig.~\ref{fig:family-nbw-ryc-smaller}.

\begin{restatable}{theorem}{comparisonCanoEq}
\label{thm:comparisonCanoEq}
Given  $u \in \finwords$, let $\canoEq$ be the congruence relation in Definition~\ref{def:congruence-rel} and $\proEq_{u}$ the right congruence in Definition~\ref{def:pro-right-congruence}.
There is a family of NBWs $\setnocond{\B_{n}}_{n \in \naturals}$ with $n+3$ states for which $\size{\canoEq} \geq n!$ and $\size{\proEq_{u}} \leq (n+3) + 2$.
\end{restatable}

\begin{figure}[t]
    \centering
    \includegraphics[scale=0.8]{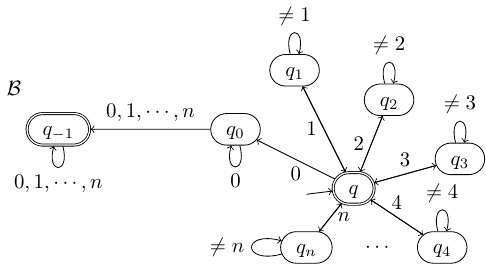}
    \caption{The family of NBWs $\setnocond{\B_{n}}_{n \in \naturals}$ over the alphabet $\setnocond{0, 1, \cdots, n }$ with $n+3$ states for which $\size{\canoEq}$ is at least $n!$ while $\size{\proEq_{u}}$ is at most $(n+3) + 2$ for each $u \in \finwords$; the initial state is $q$ and $\acc = \setnocond{q, q_{-1}}$.
    We remark that this NBW is inspired by a DBW from~\cite{AngluinF16}. However, our NBW is not deterministic.}
    \label{fig:family-nbw-ryc-smaller}
\end{figure}

The idea underlying this result is that in $\icanoEq$, there are at most $n+4$ equivalence classes, which correspond to the singletons $\setnocond{r}$ with $r \in \states$ and the set $\setnocond{q_{-1},q_{0}}$.
For each of these classes, say $\iclass{u} = \iclass{1}$, the associated classes $[v]_{\proEq_{u}}$ can correspond to at most $(n+3)+2$ configurations of (accepting) runs, like $\setnocond{q_{1}\pathto{v}{} q_{1}}$ and $\setnocond{q_{1}\pathto{v}{} q_{0}, q_{1}\fpathto{v}{} q_{0}, q_{1}\pathto{v}{} q_{-1}, q_{1}\fpathto{v}{} q_{-1}}$.
On the other hand, since $\canoEq$ must take care of both prefixes and periods, different permutations of $\setnocond{1, \cdots, n}$ taken as prefixes cannot be equivalent, thus $\size{\canoEq} \geq n!$.

When working with DBWs, the overall index of the right congruence relations $\bigcup_{u \in \finwords} \setnocond{\proEq_{u}}$ can be exponentially better than $\canoEq$ (cf.~Lemma~\ref{lem:lower-bound-congrel}).

\begin{restatable}{theorem}{dbwSizeImprovedFDFW}
\label{thm:size-dbw-improved-fdfw}
Let $\A$ be a DBW with $n$ states.
Then $\Sigma_{\iclass{u} \in \iquotient} \size{\proEq_{u}} \in \bigO(n^{2})$.
\end{restatable}

Similarly to Lemma~\ref{lem:saturation-congruence}, the saturation lemma for $\canoEq$, the right congruences $\icanoEq$ and $\proEq_{u}$ with $u \in \finwords$ also allow us to recognize \emph{exactly} $\lang{\A}$ or its complement $\infwords \setminus \lang{A}$:
for these relations we have again that the $\omega$-language $\iclass{u} \proclass{v}^{\omega}$ with $uv \icanoEq u$ is either completely inside $\lang{\A}$ or outside $\lang{\A}$, even if we drop the requirement $\proclass{v} \proclass{v} \subseteq \proclass{v}$.

\begin{restatable}[Saturation Lemma for $(\icanoEq, \cup_{u\in\finwords} \setnocond{\proEq_{u}})$]{lemma}{saturationLemmaImproved}
\label{lem:saturation-impoved-fdfws}
\begin{enumerate}
\item For $u \in \finwords, v \in \poswords$, if $uv\icanoEq u$, then either $\iclass{u}\proclass{v}^{\omega} \cap \lang{\A} = \emptyset$ or $\iclass{u}\proclass{v}^{\omega} \subseteq \lang{\A}$.

\item $\infwords = \bigcup \setcond{\iclass{u}\proclass{v}^{\omega}}{u \in \finwords, v \in \poswords, uv\icanoEq u}$.
\item $\infwords \setminus \lang{\A} = \bigcup \setcond{\iclass{u}\proclass{v}^{\omega}}{u \in \finwords, v \in \poswords, uv\icanoEq u, \iclass{u}\proclass{v}^{\omega} \cap \lang{\A} = \emptyset}$.
\end{enumerate}
\end{restatable}
By definition of $\icanoEq$ and $\proEq_{u}$, if $uv \icanoEq u$, then the set of states $\trans(\inits, u)$ is visited infinitely often when reading the word $w = uv^{\omega}$;
it also implies $uv^{j} \icanoEq u$ for each $j \geq 0$.
Moreover, if $w \in \lang{\A}$, then there is a run of $\A$ over $w$ that is accepting, i.e., the run visits infinitely often states in $\acc$.
This happens when dealing with $v^{\omega}$, since $u$ is a finite word;
thus $\A$ visits an accepting state when reading $v$ on the way from $\trans(\inits, u)$ to $\trans(\trans(\inits, u), v) = \trans(\inits, u)$.

These properties allow us to prove Item~(1), i.e., that if $w \in \iclass{u}\proclass{v}^{\omega} \cap \lang{\A}$, then for each word $w' \in \iclass{u}\proclass{v}^{\omega}$ we have $w' \in \lang{\A}$, because $w'$ can be written as $w' = u' \cdot v'_{1} \cdot v'_{2} \cdots$ with $u' \in \iclass{u}$ and $v'_{j} \in \proclass{v}$ for each $j \geq 1$.
Thus for $w'$ we have that $\A$ visits infinitely often the set $\trans(\inits, u'v'_{j}) = \trans(\inits, u') = \trans(\inits, u)$ while visiting an accepting state on the way from $\trans(\inits, u')$ to $\trans(\trans(\inits, u'), v'_{j}) = \trans(\inits, u')$ since $v \proEq_{u} v'_{j}$.
Item~(2) holds by considering only the UP-words (cf.~Theorem~\ref{thm:upword-omega-regular-language}), for which we have that for each $w=u'v'^{\omega} \in \upword{\infwords}$, we can construct a decomposition $(u=u'v'^h, v=v'^{k})$ of $w$ with $u \icanoEq uv$ for some $h, k \geq 1$ since $\icanoEq$ is of finite index.
By combining Items~(1) and~(2), to obtain $\infwords\setminus\lang{\A}$ we can just take the union of all languages $\iclass{u}\proclass{v}^{\omega}$ such that $\iclass{u}\proclass{v}^{\omega} \cap \lang{\A} = \emptyset$.

\section{Optimal Congruence Relations for NBWs}
\label{sec:optimal-rc}

The right congruence relations we introduced in Sect.~\ref{sec:improved-ramsey-for-fdfws}, despite improving $\canoEq$, still lead to a blow-up of $3^{\bigO(n^{2})}$ (cf.~Lemma~\ref{lem:size-pro-eq}).
The main cause of the exponent $n^{2}$ is that it is possible that each of the $n$ states is a predecessor of a state $r$ over the word $v$, i.e., it reaches $r$ over the word $v$.
To avoid having to consider all such precedessors, we look for specific representatives, in order to reduce the blow-up.
Inspired by~\cite{Fogarty13,DBLP:journals/iandc/FogartyKVW15}, we introduce a \emph{preorder} on the states based on the transition structure of $\A$;
we then use such preorder to select the representatives.
In particular, if the predecessors of $r$ can be reduced to only one representative for a given $v$, we obtain that the blow-up reduces to $2^{\bigO(n\log n)}$.
The representative we are going to use is the maximal equivalence class induced by the preorder among at most $n$ equivalence classes.
Breuers \emph{et al.}~\cite{DBLP:conf/fossacs/BreuersLO12} also proposed a preorder-based optimization to improve \rbc;
see Remark~\ref{rmk:comparison} for a detailed comparison.

Inspired by~\cite{Fogarty13,DBLP:journals/iandc/FogartyKVW15}, in the remainder of this section we present a preorder $\preceq_{u}$ on the set of states $\trans(\inits, u)$, $u \in \finwords$, yielding optimal relations for NBWs.
To that end, we first show how to compare the finite runs of $\A$ over a given word.

Fix a finite word $u$;
given a run $\pi$ of $\A$ over $u$, recall that $\wordletter{\pi}{i}$ denotes the $i$-th element (i.e. state) of $\pi$.
For each run $\pi$, we can define the function $\dirac(\pi) = \dirac(\pi[1]) \dirac(\pi[2]) \dots$  where $\dirac(s) = 1$ if $s \in \acc$ and $0$ if $s \notin \acc$.
In other words, each run $\pi$ can be encoded as a binary sequence.
Given two runs $\pi, \pi'$ of $\A$ over $u$, we say that $\pi'$ is \emph{greater} than $ \pi$, denoted by $\pi' > \pi$,  if there is a prefix $\alpha 1$ of $\dirac(\pi')$ such that $\alpha 0$ is a prefix of $\dirac(\pi)$.
That is, $\pi'$ is greater than $ \pi$ if there is an integer $1 \leq j \leq \size{u}+1$ such that $\pi'[j] \in \acc$, $\pi[j] \notin \acc$ and $\pi'[i] \in \acc \Longleftrightarrow \pi[i] \in \acc$ for all $1 \leq i < j$ ($i$ does not exist when $j = 1$).

Let $\Pi_{u,q}$ be the set of runs of $\A$ over $u$ starting from $\inits$ with the last state being $q$.
For each state $q \in \trans(\inits, u)$, there may be several runs in $\Pi_{u, q}$;
the set of \emph{maximal} runs in $\Pi_{u, q}$ is defined as $\max(\Pi_{u, q}) = \setcond{\pi' \in \Pi_{u,q}}{\forall \pi \in \Pi_{u,q}, \pi \not> \pi'}$.
The following result is a direct consequence of the definition above.
\begin{proposition}
\label{prop:eq-max-runs}
For two runs $\pi_{q}, \pi'_{q} \in \max(\Pi_{u, q})$, for each $1 \leq i \leq \size{u}+1$ we have that $\pi_{q}[i] \in \acc \Longleftrightarrow \pi'_{q}[i] \in \acc$.
\end{proposition}
That is, all runs in $\max(\Pi_{u,q})$ have the same image under $\dirac$.

In the following, we define the preorder $\preceq_{u}$ on the set of states $P = \trans(\inits, u)$ by comparing the sets of maximal runs $\max(\Pi_{u, q})$ and $\max(\Pi_{u, r})$ for $q, r \in P$.
\begin{definition}[Preorder $\preceq_{u}$]
\label{def:preorder-u}
Given $u \in \finwords$,
\begin{itemize}
\item
    if $u = \emptyword$, then for $q, r \in \inits = \trans(\inits, u)$, we define $q \preceq_{\emptyword} r$ iff $q \in \acc$ implies $r \in \acc$.
    Therefore $q \prec_{\emptyword} r$ iff $q \notin \acc$ and $r \in \acc$;
\item
    when $u \in \poswords$, for $q, r \in \trans(\inits, u)$, $q \preceq_{u} r$ if the runs $\pi_{q} \in \max(\Pi_{u,q})$ are not greater than the runs $\pi_{r} \in \max(\Pi_{u,r})$.
    In particular, $q \prec_{u} r$ if the runs $\pi_{r} \in \max(\Pi_{u,r})$ are greater than the runs $\pi_{q} \in \max(\Pi_{u,q})$.
\end{itemize}
\end{definition}

One can verify that $\preceq_{u}$ is a binary relation that is reflexive (i.e., for each $q \in \states$, $q \preceq_{u} q$) and transitive (i.e., for each $q,r,s \in \states$, $q \preceq_{u} r$ and $r \preceq_{u} s$ implies $q \preceq_{u} s$), so it is a preorder;
we also have $q \prec_{u} r$ whenever $q \preceq_{u} r$ and $r \not \preceq_{u} q$ and we write $q \simeq_{u} r$ whenever $q \preceq_{u} r$ and $r \preceq_{u} q$.
Intuitively, we have $q \prec_{u} r$ if there is a run from an initial state to $r$ on $u$ that sees an accepting state earlier than all paths from the initial states to $q$ on $u$.
That is, there is a prefix $\alpha 1$ of $\dirac(\pi_{r})$ for a run $\pi_{r}$ to $r$ such that $\alpha 0$ is a prefix of $\dirac(\pi_{q})$ for all runs $\pi_{q}$ to $q$.

Due to Proposition~\ref{prop:eq-max-runs}, if there is a run $\pi_{r} \in \max(\Pi_{u, r})$ greater than a run in $\max(\Pi_{u,q})$, then all runs in $\max(\Pi_{u,r})$ are greater than the runs in $\max(\Pi_{u,q})$.

\begin{example}
\label{ex:preorder}
As an example, consider the NBW $\B_{n}$ depicted in Fig.~\ref{fig:family-nbw-ryc-smaller} and let $P = \trans(\setnocond{q}, 00) = \setnocond{q_{-1}, q_{0}}$;
we have $q_{0} \prec_{00} q_{-1}$ on this set since there is a run from the initial state $q$ to $q_{-1}$ that sees the accepting state $q_{-1}$ after inputting the second $0$ while all runs from $q$ to $q_{0}$ on $00$ do not visit an accepting state right after inputting the second $0$.
\end{example}

\begin{remark}
The preorder $\preceq_{u}$ in Definition~\ref{def:preorder-u} shares the same idea of comparing the maximal runs with the \emph{lexicographical order} of vertices at the same level of the run direct acyclic graph (DAG) over an $\omega$-word $w$ used in~\cite{Fogarty13};
see~Appendix~\ref{app:comparisonProfile} for detailed connection between their and our works.
The difference between our work and the work in~\cite{Fogarty13} is that the latter applies this idea to Slice-based~\cite{kahler2008complementation} and Rank-based complementation algorithms~\cite{kupferman2001weak} while ours is designed for \rbc.
A similar idea was also used in~\cite{DBLP:journals/iandc/FogartyKVW15} for determinizing NBWs.
\end{remark}

As an immediate consequence of Definition~\ref{def:preorder-u}, given two states $q, r \in \trans(\inits, u)$ such that $q \preceq_{u} r$, we have that a run from an initial state to $q$ that visits an accepting state mandates that there must be a run from an initial state to $r$ that also visits accepting states.
\begin{corollary}
\label{coro:preceq-init-to-q}
Let $q, r \in \trans(\inits, u)$, with $q \preceq_{u} r$.
Then $\iota_{q} \fpathto{u}{} q$ for some initial state $\iota_{q} \in \inits$ implies $\iota_{r} \fpathto{u}{} r$ for some initial state $\iota_{r} \in \inits$.
\end{corollary}

Let $P = \trans(\inits, u)$.
The preorder $\preceq_{u}$ defines a partition of $P$ in which states in the same set are equivalent under $\preceq_{u}$.
By abuse of terminology, we call the set $[r]_{\preceq_{u}} = \setcond{r' \in P}{r ' \simeq_{u} r}$ the equivalence class of $r \in P$ under $\preceq_{u}$;
we denote by $P/_{\preceq_{u}}$ the set of all such equivalence classes.
Since every two states $q, r \in P$ are \emph{comparable} under $\preceq_{u}$, we define the maximal equivalence class of $P$ under $\preceq_{u}$ as $\max_{\preceq_{u}}(P) = \max(P/_{\preceq_{u}})= \setcond{r \in P}{\text{$r' \preceq_{u} r$ for all $r' \in P$}}$;
moreover, the equivalence classes in $P/_{\preceq_{u}}$ can be linearly ordered by $[r]_{\preceq_{u}} \classpreq_{u} [r']_{\preceq_{u}} \Longleftrightarrow r {\preceq_{u}} r'$;
so we have $[r]_{\preceq_{u}}\classpre_{u} [r']_{\preceq_{u}}$ if $[r]_{\preceq_{u}} \classpreq_{u} [r']_{\preceq_{u}}$ and $[r'] \not \classpreq_{u} [r]_{\preceq_{u}}$.
Here $\classpreq_{u}$ is a partial order, not a preorder, which implies that $[r]_{\preceq_{u}} = [r']_{\preceq_{u}}$ iff $[r]_{\preceq_{u}} \classpreq_{u} [r']_{\preceq_{u}}$ and $[r']_{\preceq_{u}} \classpreq_{u} [r]_{\preceq_{u}}$.

An interesting property of the states $\A$ visits on the maximal runs from the initial states $\inits$ to a state $q \in \trans(\inits, uv)$ over the finite word $uv$ is that they are step by step all equivalent under the preorder w.r.t. the prefix of $uv$.
We denote by $\max(\Pi_{uv, q})|_u$ the set $\setcond{\pi[\size{u} + 1]}{\pi \in \max(\Pi_{uv, q})}$, i.e., the set of states reached from initial states after inputting $u$ on the maximal runs to $q$ over $uv$.

\begin{restatable}{lemma}{subsetOfMaxPreds}
\label{lem:subseteq-class-u}
Given $u, v \in \finwords$ and $q \in \trans(\inits, uv)$, let $[p]_{\preceq_{u}} = \max_{\preceq_{u}} \setcond{[p']_{\preceq_{u}} \in \trans(\inits, u)/_{\preceq_{u}}}{p' \pathto{v}{} q}$.
Then for each $q' \in [q]_{\preceq_{uv}}$, $\max(\Pi_{uv, q'})|_u \subseteq [p]_{\preceq_{u}}$.
\end{restatable}

By definition of $[q]_{\preceq_{uv}}$, all the maximal runs from the initial states to the states in $[q]_{\preceq_{uv}}$ have the same image under $\dirac$.
As a consequence, the states on these runs reached after reading $u$ must belong to the same equivalence class $[p]_{\preceq_{u}}$ under $\preceq_{u}$, which is also the maximal equivalence class  under $\preceq_{u}$ that reaches $[q]_{\preceq_{uv}}$.
If this would not be the case, then we would be able to find runs to $[q]_{\preceq_{uv}}$ greater than the current maximal runs by visiting $[p]_{\preceq_{u}}$.

A useful property of these maximal runs is that they share visits to accepting states;
more precisely, if one of the maximal runs on a word $uv$ to a state $q_{1} \in [q]_{\preceq_{uv}}$ visits an accepting state while reading $v$, then all other maximal runs on the same word to some other state $q_{2} \in [q]_{\preceq_{uv}}$ do.
The motivation for this is again the maximality of the runs: if one run visits an accepting state while another does not, then the former is greater than the latter, which implies that the latter cannot be maximal.
This property is formalized below.

\begin{restatable}{lemma}{saturationAccPreorder}
\label{lem:all-or-none-acc-preorder}
Let $u, v \in \finwords$ and $q \in \trans(\inits, uv)$.
For $q_{1}, q_{2} \in [q]_{\preceq_{uv}}$, $p_{1} \in \max(\Pi_{uv, q_{1}})|_u$ and $p_{2} \in \max(\Pi_{uv, q_{2}})|_u$, it holds that $p_{1} \fpathto{v}{} q_{1} $ iff $ p_{2} \fpathto{v}{}q_{2}$.
\end{restatable}
Similarly to Lemma~\ref{lem:subseteq-class-u}, a consequence of  Lemma~\ref{lem:all-or-none-acc-preorder} is that, for a given finite word $u$ and $q_{1} \simeq_{u} q_{2}$, the maximal runs in $\max(\Pi_{u, q_{1}})$ and in $\max(\Pi_{u, q_{2}})$ visit accepting states at the same moment, i.e., they have the same image under $\dirac$.

The preorder $\preceq_{u}$ enjoys several properties about the states and maximal runs of $\A$ for the given finite word $u$.
Thus, instead of tracing only the set of reachable states, as done by the right congruence $\icanoEq$ (cf.~Definition~\ref{def:right-congruence}), we also trace the reachable states $\trans(\inits, u)$ with the preorder $\preceq_{u}$ to get the right congruence $\canoEq^{o}$.

\begin{definition}[RC $\canoEq^{o}$]
\label{def:opt-canoEq}
For $u_{1}, u_{2} \in \finwords$, we say $u_{1} \canoEq^{o} u_{2}$ iff $\trans(\inits, u_{1})/_{\preceq_{u_{1}}} = \trans(\inits, u_{2})/_{\preceq_{u_{2}}}$.
\end{definition}
Consider again the NBW $\B_{n}$ in Fig.~\ref{fig:family-nbw-ryc-smaller}:
we can represent $\trans(\inits, 00)/_{\preceq_{00}}$ as an ordered sequence of sets $\langle \setnocond{q_{0}}, \setnocond{q_{-1}}\rangle$ since we have $\setnocond{q_{0}} \classpre_{00} \setnocond{q_{-1}}$.
Analogously, $\trans(\inits, 000)/_{\preceq_{000}}$ can also be represented as $\langle \setnocond{q_{0}}, \setnocond{q_{-1}}\rangle$ while $\trans(\inits, 001)/_{\preceq_{001}}$ as $\langle \setnocond{q_{-1}}\rangle$.
We can see that $00 \canoEq^{o} 000$ since $\trans(\inits, 00)/_{\preceq_{00}} = \trans(\inits, 000)/_{\preceq_{000}} = \langle \setnocond{q_{0}}, \setnocond{q_{-1}} \rangle$ while $000 \not\canoEq^{o} 001$ as $\trans(\inits, 001)/_{\preceq_{001}} = \langle \setnocond{q_{-1}} \rangle$.

Since each equivalence class $\oclass{u}$, $u \in \finwords$, can be uniquely encoded as the set $\trans(\inits, u)/_{\preceq_{u}}$, i.e., an ordered sequence of sets over $\states$,
by~\cite{Fogarty13} we have that the number of possible ordered sequences of sets over $\states$ is $\bigO((\frac{n}{e \ln n})^{n}) \approx (0.53n)^{n} \leq n^{n}$.
Thus we have the following upper bound for $\canoEq^{o}$, so it is of finite index.

\begin{restatable}{lemma}{sizeOfOptCanoEq}
\label{lem:size-opt-canoEq}
Let $\canoEq^{o}$ be the right congruence in Definition~\ref{def:opt-canoEq}.
Then $\size{\canoEq^{o}} \leq n^{n}$.
\end{restatable}


Given their definitions, it is clear that $\icanoEq$ is coarser than $\canoEq^{o}$, thus $\size{\canoEq^{i}} \leq \size{\canoEq^{o}}$.
Nonetheless, the right congruence $\canoEq^{o}$ allows us to define a novel right congruence relation $\proEq^{o}_{u}$ of index $2^{\bigO(n \log n)}$, for a given $u \in \finwords$.

\begin{definition}[RC $\proEq^{o}_{u}$]
\label{def:opt-proEq}
Given $u, v_{1},v_{2} \in \finwords$, we say $v_{1} \proEq^{o}_{u} v_{2}$ if and only if (1) $uv_{1} \canoEq^{o} uv_{2}$, and (2) for all states $q \in P'$, for $S_{1} = \max_{\preceq_{u}}\setcond{[p]_{\preceq_{u}} \in P/_{\preceq_{u}}}{p \pathto{v_{1}}{}q} $ and $S_{2} = \max_{\preceq_{u}}\setcond{[p]_{\preceq_{u}} \in P/_{\preceq_{u}}}{p \pathto{v_{2}}{}q}$, we have (i) $S_{1} = S_{2}$ and (ii) $p_{1} \fpathto{v_{1}}{}q$ for $p_{1} \in \max(\Pi_{uv_{1},q})|_{u}$ iff $ p_{2} \fpathto{v_{2}}{}q$ for some $p_{2} \in \max(\Pi_{uv_{2},q})|_{u}$ where $P = \trans(\inits, u)$ and $P' = \trans(\inits, uv_{1}) = \trans(\inits, uv_{2})$.
\end{definition}
Note that the equality $\trans(\inits, uv_{1}) = \trans(\inits, uv_{2})$ holds because, under the assumption $uv_{1} \canoEq^{o} uv_{2}$, by definition of $\canoEq^{o}$, we have that the sets of equivalence classes $\trans(\inits, uv_{1})/_{\preceq_{u_{1}}}$ and $\trans(\inits, uv_{2})/_{\preceq_{u_{2}}}$ are equal, so must be $\trans(\inits, uv_{1})$ and $\trans(\inits, uv_{2})$.

Definition~\ref{def:opt-proEq} formalizes the following idea for recognizing the $\omega$-words accepted and rejected by $\A$.
Since we want to use $(\canoEq^{o}, \cup_{u \in \finwords} \setnocond{\proEq^{o}_{u}})$ to characterize $\lang{\A}$ and $\infwords \setminus \lang{\A}$, i.e., to establish its saturation lemma in line with $\canoEq$ (cf.~Lemma~\ref{lem:saturation-congruence}) and $(\icanoEq, \cup_{u \in \finwords} \setnocond{\proEq_{u}})$ (cf.~Lemma~\ref{lem:saturation-impoved-fdfws}), under the assumption that $uv_{1} \canoEq^{o} u $ and $u \canoEq^{o} uv_{2}$, we need to guarantee that if $v_{1} \proEq^{o}_{u} v_{2}$, then $uv_{1}^{\omega} \in \lang{\A} $ if and only if $uv_{2}^{\omega} \in \lang{\A}$.
To achieve this, the first condition we impose (cf. (1) of Definition~\ref{def:opt-proEq}) is to visit infinitely often the same states over the $\omega$-words $uv^{\omega}_{1}$ and $uv^{\omega}_{2}$;
so we require $uv_{1} \canoEq^{o} uv_{2}$.
The second condition is to guarantee that the maximal runs over $uv_{1}^{k}$ and $uv_{2}^{k}$, $k \geq 1$, share the same image under $\dirac$;
so when extending to infinite words, the image under $\dirac$ will still be the same.
This ensures that $uv_{1}^{\omega} \in \lang{\A}$ if and only if $uv_{2}^{\omega} \in \lang{\A}$.
To guarantee to have the same image, we first require that the maximal equivalence classes from each state $q \in \trans(\inits, u)$ over both finite words $v_{1}$ and $v_{2}$ have to be the same (cf. condition (2)-(i), together with Lemma~\ref{lem:subseteq-class-u});
then, we demand that they share the visits to accepting states (cf. condition (2)-(ii) with Lemma~\ref{lem:all-or-none-acc-preorder}).

Consider again Example~\ref{ex:preorder} and let $u = \emptyword, v_{1} = 00$ and $v_{2} = 000$;
we want to check whether $00 \proEq^{o}_{\emptyword} 000$.
Clearly $\emptyword \cdot 00 \canoEq^{o} \emptyword \cdot 000$ since $00 \canoEq^{o} 000$.
$\trans(\inits, \emptyword)/_{\preceq_{\emptyword}}$ can be represented with $\langle \setnocond{q} \rangle$, a singleton;
so obviously, we have $S_{1} = S_{2} = \setnocond{q}$, thus we satisfy Condition (2)-(i) of Definition~\ref{def:opt-proEq}.
To fulfill Condition (2)-(ii), we first have $P' = \setnocond{q_{-1}, q_{0}}$.
We also have $\max(\Pi_{\emptyword \cdot 00, q_{-1}}) = \setnocond{qq_{0}q_{-1}}$, $\max(\Pi_{\emptyword \cdot 00,q_{0}}) = \setnocond{qq_{0}q_{0}}$, $\max(\Pi_{\emptyword \cdot 000, q_{-1}}) = \setnocond{qq_{0}q_{-1}q_{-1}}$ and $\max(\Pi_{\emptyword \cdot 000,q_{0}}) = \setnocond{qq_{0}q_{0}q_{0}}$.
For state $q_{-1} \in P'$, (2)-(ii) is satisfied since $qq_{0}q_{-1}$ and $qq_{0}q_{-1}q_{-1}$ both visit accepting states.
For state $q_{0} \in P'$, (2)-(ii) is also fulfilled as $qq_{0}q_{0}$ and $qq_{0}q_{0}q_{0}$ both visit accepting state $q$.
Therefore, we conclude that $00 \proEq^{o}_{\emptyword} 000$ holds.
Clearly $000 \not\proEq^{o}_{\emptyword} 001$ since we already know that $\emptyword \cdot 000 \not\canoEq^{o} \emptyword \cdot 001$.

As desired before Definition~\ref{def:opt-proEq}, the index of $\proEq^{o}_{u}$ is indeed in $2^{\bigO(n \log n)}$.
\begin{restatable}[]{lemma}{sizeOfOptProEq}
\label{lem:upper-bound-opt-proEq}
Given $u \in \finwords$, let $\proEq^{o}_{u}$ be the right congruence from Definition~\ref{def:opt-proEq}.
Then $\size{\proEq^{o}_{u}} \leq n^{n} \times (n+1)^{n} \times 2^{n} \in 2^{\bigO(n \log n)}$.
\end{restatable}
The upper bound for $\size{\proEq^{o}_{u}}$ derives from the encoding we use for $[v]_{\proEq^{o}_{u}}$.
$[v]_{\proEq^{o}_{u}}$ is mapped to the pair $\langle \trans(\inits, uv)/_{\preceq_{uv}}, f \rangle$ where the function $f$ keeps track of the satisfaction of the states $q \in \states$ of the conditions in Definition~\ref{def:opt-proEq}, i.e., whether $q \in \trans(\inits, uv)$ and Conditions (2)-(i) and (2)-(ii) for the such states.
The codomain of $f$ has size $2n + 1 < 2(n+1)$, so the possible different functions $f$ are $(2(n+1))^{n} = 2^{n} \times (n+1)^{n}$, while by~\cite{Fogarty13} the possible sets $\trans(\inits, uv)/_{\preceq_{uv}}$ are $n^{n}$, hence $\size{\proEq^{o}_{u}} \leq n^{n} \times (n+1)^{n} \times 2^{n} \in 2^{\bigO(n \log n)}$.

Similarly to Lemma~\ref{thm:size-dbw-improved-fdfw}, if we restrict ourselves to DBWs, then $\size{\proEq^{o}_{u}}$ is exponentially better than the bound $2^{\bigO(n \log n)}$ we have for general NBWs.
\begin{restatable}{lemma}{dbwSizeOptFDFW}
\label{lem:size-dbw-opt-fdfw}
Let $\A$ be a DBW with $n$ states.
Then $\Sigma_{\oclass{u} \in \oquotient} \size{\proEq^{o}_{u}} \in \bigO(n^{2})$.
\end{restatable}
This result follows from the fact that, being $\A$ deterministic, then there are at most $n$ classes $\oclass{u} \in \oquotient$;.
By taking the same encoding as in Lemma~\ref{lem:upper-bound-opt-proEq}, this time the index of $\proEq^{o}_{u}$ is at most $2n$, so the result follows.

Similarly to the other (right) congruence relations we considered, i.e., $\canoEq$ (cf.~Lemma~\ref{lem:saturation-congruence}) and $(\icanoEq, \cup_{u \in \finwords} \setnocond{\proEq_{u}})$ (cf.~Lemma~\ref{lem:saturation-impoved-fdfws}), also $(\canoEq^{o}, \cup_{u\in\finwords} \setnocond{\proEq^{o}_{u}})$ enjoys its saturation lemma.
As stated below, $(\canoEq^{o}, \cup_{u \in \finwords} \setnocond{\proEq^{o}_{u}})$ is able to recognize exactly $\lang{\A}$ and $\infwords \setminus \lang{\A}$;
a core property to obtain this is again that the $\omega$-languages $\oclass{u}\oproclass{v}^{\omega}$ are included either in $\lang{\A}$ or in its complement $\infwords \setminus \lang{A}$.

\begin{restatable}[Saturation Lemma for $(\canoEq^{o}, \cup_{u\in\finwords} \setnocond{\proEq^{o}_{u}})$]{lemma}{optSaturationLemma}
\label{lem:saturation-opt-fdfws}
\begin{enumerate}
\item For $u \in \finwords, v \in \poswords$, if $uv\canoEq^{o} u$, then either $\oclass{u}\oproclass{v}^{\omega} \cap \lang{\A} = \emptyset$ or $\oclass{u}\oproclass{v}^{\omega} \subseteq \lang{\A}$.

\item $\infwords = \bigcup \setcond{\oclass{u}\oproclass{v}^{\omega}}{u \in \finwords, v \in \poswords, uv\canoEq^{o} u}$.
\item $\infwords \setminus \lang{\A} = \bigcup \setcond{\oclass{u}\oproclass{v}^{\omega}}{u \in \finwords, v \in \poswords, uv\canoEq^{o} u, \oclass{u}\oproclass{v}^{\omega} \cap \lang{\A} = \emptyset}$.
\end{enumerate}
\end{restatable}
The proof for this saturation lemma follows the same steps as for the other two saturation lemmas, with the appropriate adaptations that take into consideration the differences in the definitions of the right congruences.

\begin{remark}
\label{rmk:comparison}
In their work~\cite{sistla1987complementation}, Sistla \emph{et al.} constructed an NBW $\B_{u,v}$ for each proper language $Y_{u, v} = \class{u}\class{v}^{\omega}$ such that $Y_{u, v} \cap \inflang{\A} = \emptyset$.
Each $\B_{u,v}$ can be constructed with two copies of the DFW $\M[\canoEq]$ induced by $\canoEq$ (cf.~Definition~\ref{def:induced-dfw}) where the first copy processes the finite prefix $u$ while the second copy is modified to accept the word $v^{\omega}$.
According to~\cite{sistla1987complementation}, the resulting NBW $\A^{c}$ has $3^{\bigO(n^{2})}$ states.
Breuers \emph{et al.}~\cite{DBLP:conf/fossacs/BreuersLO12} also proposed a subset construction for improving \rbc for complementing NBWs;
in particular, they used the subset construction to process the finite prefix $u$ of a UP-word $uv^{\omega}$ in $\infwords \setminus \lang{\A}$.
On the other hand, they still used the classical congruence relation $\canoEq$ for recognizing the periodic word $v$ of $uv^{\omega}$.
Differently from the algorithms proposed in~\cite{sistla1987complementation,DBLP:conf/fossacs/BreuersLO12}, we exploit the right congruence $\proEq_{u}$ or $\proEq^{o}_{u}$ instead of the congruence $\canoEq$ for accepting the period $v$ of $uv^{\omega}$;
this can result in a considerable decrease of the index of the relation  (cf.~Theorem~\ref{thm:comparisonCanoEq}), which influences the number of states of the automata we build from these relations.
The part for accepting $v$ in~\cite{DBLP:conf/fossacs/BreuersLO12} has also been optimized with a preorder and its size is also reduced to $2^{\bigO(n \log n)}$.
While leading to the same upper bound, there is a difference on the automata needed to process the period $v$:
for a given $u$, the construction given in~\cite{DBLP:conf/fossacs/BreuersLO12} uses more than one automaton for recognizing $v$;
instead, our approach needs one automaton, because the equivalence class $\oclass{u}$ of $\canoEq^{o}$ only relates with one right congruence relation $\proEq^{o}_{u}$.
This allows us to represent $(\canoEq^{o}, \cup_{u\in\finwords} \setnocond{\proEq^{o}_{u}})$ as an FDFW, as we explain in Section~\ref{sec:applications}.
\end{remark}

\section{Connection to FDFWs}
\label{sec:applications}

In this section, we highlight the deep connection between the congruence relations of NBWs and FDFWs.
This allows us to use the right congruences $(\mathord{\canoEq^{o}}, \bigcup_{u \in \finwords} \setnocond{\mathord{\proEq^{o}_{u}}})$ we introduced in Section~\ref{sec:optimal-rc} to construct an FDFW $\F$ with \emph{optimal} complexity that accepts $\infwords \setminus \lang{\A}$.
As a byproduct of this connection, we are able to prove Theorem~\ref{thm:lower-bound-rc};
in other words, one cannot find congruence relations for NBWs of index less than $2^{\bigO(n \log n)}$.

We now introduce the construction of FDFWs from the right congruences.
Since $\canoEq^{o}$ (resp., $\icanoEq$) and $\proEq^{o}_{u}$ (resp., $\proEq_{u}$) with $u \in \finwords$ are right congruences of finite index, by means of Definition~\ref{def:induced-dfw} they can be used to define the transition structures of the DFWs of an FDFW $\F$ recognizing $\infwords \setminus \lang{\A}$.
Moreover, by Lemma~\ref{lem:saturation-opt-fdfws} (resp., Lemma~\ref{lem:saturation-impoved-fdfws}), we can identify the accepting macrostates of the progress DFWs.
We now give the construction of the FDFW $\F$ with $\canoEq^{o}$ and $\proEq^{o}_{u}$.
The construction of the FDFW with $\icanoEq$ and $\proEq_{u}$ is similar.

\begin{definition}
\label{def:improved-ramsey-fdfws}
The FDFW $\F$ is a tuple $(\M[\canoEq^{o}], \setnocond{\N_{u}[\proEq^{o}_{u}]})$ where
\begin{itemize}
\item
    $\M[\canoEq^{o}]$ is the DFW induced by $\canoEq^{o}$ according to Definition~\ref{def:induced-dfw};
\item
    for each macrostate $\oclass{u}$ of $\M[\canoEq^{o}]$, the progress DFW $\N_{u}[\proEq^{o}_{u}]$ is constructed as in Definition~\ref{def:induced-dfw} parameterized with $\proEq^{o}_{u}$.
    The accepting macrostates of $\N_{u}[\proEq^{o}_{u}]$ are the equivalence classes $\oproclass{v}$ of $\proEq^{o}_{u}$ such that $uv \canoEq^{o} u$ and $\oclass{u} \oproclass{v}^{\omega} \cap \lang{\A} = \emptyset$.
\end{itemize}
\end{definition}

The FDFW constructed according to Definition~\ref{def:improved-ramsey-fdfws} has the desired properties we are looking for:
it accepts $\infwords \setminus \lang{\A}$ and has only $2^{\bigO(n \log n)}$ macrostates.
\begin{restatable}{theorem}{correctnessOfImprovedFDFW}
\label{thm:recurrent-fdfw-neg}
Let $\F$ be the FDFW constructed from $\A$ in Definition~\ref{def:improved-ramsey-fdfws}.
Then
(1) $\upword{\F} = \upword{\infwords \setminus \inflang{\A}}$;
(2) $\F$ is saturated; and
(3) $\F$ has $2^{\bigO(n \log n)}$ macrostates.
\end{restatable}

The three results stated in Theorem~\ref{thm:recurrent-fdfw-neg} follow by the definition of $\F$ and the properties of $(\mathord{\canoEq^{o}}, \bigcup_{u \in \finwords} \setnocond{\mathord{\proEq^{o}_{u}}})$:
result (1) is a direct consequence of Definition~\ref{def:improved-ramsey-fdfws};
result (2) is implied by the saturation lemma for $(\mathord{\canoEq^{o}}, \bigcup_{u \in \finwords} \setnocond{\mathord{\proEq^{o}_{u}}})$ (cf. Item (1) of~Lemma~\ref{lem:saturation-opt-fdfws});
and result (3) by the indexes of $(\mathord{\canoEq^{o}}, \bigcup_{u \in \finwords} \setnocond{\mathord{\proEq^{o}_{u}}})$ and the construction of the DFWs of $\F$.

We are now able to formalize the optimality of our FDFW construction of $\F$ based on $(\mathord{\canoEq^{o}}, \bigcup_{u \in \finwords} \setnocond{\mathord{\proEq^{o}_{u}}})$.
The upper bound is due to Theorem~\ref{thm:recurrent-fdfw-neg};
the matching lower bound comes from the well-known fact~\cite{Yan/08/lowerComplexity} that there exists a family of NBWs $\setnocond{\A_{n}}_{n \in \naturals}$ whose complementary NBWs $\setnocond{\A_{n}^{c}}_{n \in \naturals}$ have $2^{\Omega(n\log n)}$ states, so the same lower bound must hold for FDFWs since there are polynomial-time translations from FDFWs to NBWs~(cf.~Lemma~\ref{lem:fdfw-to-nbw}).

\begin{theorem}
\label{thm:lower-bound}
The construction of FDFWs in Definition~\ref{def:improved-ramsey-fdfws} with the right congruence relations $(\mathord{\canoEq^{o}}, \bigcup_{u \in \finwords} \setnocond{\mathord{\proEq^{o}_{u}}})$ from $\A$ is asymptotically optimal.
\end{theorem}

\begin{remark}
Here we discuss related works on FDFWs.
As mentioned before, there are polynomial-time translations from FDFWs to NBWs~\cite{calbrix1993ultimately, AngluinBF18}.
The opposite translation is more challenging:
the direct translations from an $n$-states NBW, proposed in~\cite{calbrix1993ultimately} and in~\cite{Kuperberg19dlt}, produce an FDFW with $\bigO(4^{n^{2} + n})$ states and an FDFW with $\bigO(3^{n^{2} + n})$ states, respectively.
Our construction in Definition~\ref{def:improved-ramsey-fdfws} replaced with $(\canoEq^{i}, \bigcup_{u \in \finwords} \setnocond{\proEq_{u}})$ can even be exponentially better than these two translations;
due to lack of space, detailed reasoning can be found in Appendix~\ref{app:comparison}.
The translation based on an intermediate determinization of NBWs to deterministic Parity automata given in~\cite{AngluinBF18} also yields an FDFW with the optimal complexity $2^{\bigO(n \log n)}$.
Our construction (cf.~Definition~\ref{def:improved-ramsey-fdfws}), however, is the \emph{first direct} and \emph{optimal} translation from an NBW to an FDFW without involving determinization of NBWs.
Like in~\cite{Kuperberg19dlt}, our congruence relation-based translation also reveals that FDFWs are actually being applied in inclusion checking of NBWs (cf. \cite{DBLP:conf/tacas/FogartyV10, DBLP:conf/cav/AbdullaCCHHMV10,DBLP:conf/concur/AbdullaCCHHMV11}).
The specialized translation for DBWs proposed in~\cite{AngluinBF18} can be used to convert a DBW $\A$ with $n$ states to a saturated FDFW $\F'$ with $n + n \times 2n \in \bigO(n^{2})$ macrostates such that $\upword{\F'} = \upword{\infwords\setminus\lang{\A}}$;
we remark that our construction for FDFWs in Sect.~\ref{sec:applications} degenerates to their construction when the given NBW $\A$ is deterministic.
Given an $\omega$-regular language $L$, Angluin and Fisman~\cite{AngluinF16} directly operate on the language $L$ and give congruence relations for constructing FDFWs of $L$.
In contrast, our work takes an NBW $\A$ as input and defines congruence relations for recognizing $\infwords \setminus \lang{\A}$ based on the transitions of $\A$.
\end{remark}

We now formalize the main result of this paper as Theorem~\ref{thm:lower-bound-rc}, that is a direct consequence of Theorem~\ref{thm:lower-bound} since the constructed FDFW has the same number of macrostates as the index the congruence relations $(\mathord{\canoEq^{o}}, \bigcup_{u \in \finwords} \setnocond{\mathord{\proEq^{o}_{u}}})$ (cf.~Definition~\ref{def:improved-ramsey-fdfws}).

\begin{restatable}{theorem}{lowerBoundOptRc}
\label{thm:lower-bound-rc}
The right congruence relations $(\canoEq^{o}, \bigcup_{u \in \finwords}\setnocond{\proEq^{o}_{u}})$ given in Definitions~\ref{def:opt-canoEq} and~\ref{def:opt-proEq}, respectively, are asymptotically optimal among all right congruence relations $(\canoEq, \bigcup_{u \in \finwords}\setnocond{\proEq_{u}})$ such that for each $u \in \finwords$ and $v \in \poswords$, if $uv \canoEq u$, then either $\class{u}\proclass{v}^{\omega} \cap \lang{\A} = \emptyset$ or $\class{u}\proclass{v}^{\omega} \subseteq \lang{\A}$.
\end{restatable}

\section{Concluding Remarks}
\label{sec:conclucion}

In this work, we have proposed coarser congruence relations than the classical congruence relation and further given asymptotically \emph{optimal} right congruences for NBWs.
Moreover, to the best of our knowledge, we give the \emph{first direct} translation from an NBW to an FDFW with \emph{optimal} complexity, based on the optimal right congruences for NBWs.
We showed that congruence relations relate tightly the classical \rbc and FDFWs.
Congruence relations are known to be able to yield the minimal DFWs for given regular languages by the Myhill-Nerode Theorem, by identifying equivalent states.
We conjecture that the resulting congruence relations
above may enable the reduction of state space in the complementary automaton of $\A$.
That is, similar subsumption and simulation techniques developed in~\cite{DBLP:conf/cav/AbdullaCCHHMV10, DBLP:conf/concur/AbdullaCCHHMV11} for the classical congruence relation can also be exploited to avoid exploration of redundant states in the containment checking between NBWs when dealing with the right congruences proposed in this work.
We leave this conjecture to future work.

\bibliographystyle{splncs04}
\bibliography{paper}

\clearpage
\appendix

\section{Proofs}
\label{app:proofs}

\sizeOfcanoEq*
\begin{proof}
An equivalence class $\class{u}$ can be encoded as a set $S \subseteq \R = \states \times \setnocond{0, 1} \times \states$ of labelled pairs of states.
Intuitively, a tuple $(q, 0, r) \in \R$ indicates that $q \pathto{u}{\trans}r$, while $(q, 1, r)$ means $q \fpathto{u}{\trans} r$.
In particular, the equivalence class $\class{\emptyword}$ is encoded as the set $\setcond{(q, 0, q)}{q \in \states} \cup \setcond{(q, 1, q)}{q \in \acc}$.
Assume that $u_{1} \canoEq u_{2}$.
Let $S_{1} = \setcond{(q, 0, r) \in \R}{q \pathto{u_{1}}{\trans} r, q \in \states, r \in \states} \cup \setcond{(q, 1, r) \in \R}{q \fpathto{u_{1}}{\trans} r, q \in \states, r \in \states}$ and $S_{2} = \setcond{(q, 0, r) \in \R}{q \pathto{u_{2}}{\trans} r, q \in \states, r \in \states} \cup \setcond{(q, 1, r) \in \R}{q \fpathto{u_{2}}{\trans} r, q \in \states, r \in \states}$.
It is easy to verify that $S_{1} = S_{2}$.
Thus an equivalence class $\class{u}$ can be uniquely represented as a set $S \subseteq \R$.

Note the set $\R$ is finite, therefore the congruence relation $\canoEq$ is of finite index.
We thus can assume that $m = \size{\canoEq} = \size{\quotient}$.
Since there can be multiple valid pairs of states for a word $u \in \finwords$, the class  $\class{u}$ corresponds to a subset of $\R$.
There are at most $n \times n$ distinct pairs of states and for each pair $(q, r)$, we either have both $(q, 1, r) \in S$ and $(q, 0, r) \in S$, just $(q, 0, r) \in S$ or the pair is not present in $S$.
Thus we have $m = \size{\canoEq} = \size{\quotient} \leq 3^{n^{2}}$.
\end{proof}

\sizeOfCongr*
\begin{proof}
We can obtain the DBW $\C_{n}$ from the NBW $\B_{n}$ in Fig.~\ref{fig:family-nbw-ryc-smaller} by removing the accepting state $q_{-1}$ and making $q_{0}$ a sink nonaccepting state;
the DBWs $\C_{n}$ are the same automata considered in \cite[Theorem 2]{AngluinF16}.
Clearly, the resulting $\C_{n}$ is a DBW.

We borrow the proof for Theorem~\ref{thm:comparisonCanoEq} below;
we note that the original proof of \cite[Theorem 2]{AngluinF16} cannot directly be used to prove our statement.
To prove that $\size{\canoEq} \geq n!$, one can just show that for each pair of different permutation words $u = i_{1}, \cdots, i_{n}$ and $u' = i'_{1}, \cdots, i'_{n}$ of $\setnocond{1, \cdots, n}$, $u \not\canoEq u'$.
We denote by $k$ the smallest position such that $i_{k} \neq i'_{k}$ where $1 \leq k < n$.
If $k = 1$, i.e., $i_{1} \neq i'_{1}$, then $u \not\canoEq u'$ by Definition~\ref{def:congruence-rel}, since we have $q \pathto{u}{} q_{i_{1}}$ and $q \pathto{u'}{} q_{i'_{1}}$ with $q_{i_{1}} \neq q_{i'_{1}}$.
Note that $\trans_{n}$ is deterministic on words without the letter $0$.
Thus $q$ transitions to $q_{i_{1}}$ after reading $i_{1}$ and stays there for the remaining $i_{2}, \cdots, i_{n}$, i.e., $q \pathto{u}{} q_{i_{1}}$.
Similarly, we have $q \pathto{u'}{} q_{i'_{1}}$.
When $1 < k < n$, we have $i_{k-1} = i'_{k-1}$.
Analogously, $u \not\canoEq u'$ since we have $q_{i_{k-1}} \pathto{u}{} q_{i_{k}}$ while $q_{i_{k-1}} \pathto{u'}{} q_{i'_{k}}$ with $i_{k} \neq i'_{k}$.
Therefore, we conclude that $u \not\canoEq u'$.
Thus, the number of equivalence classes of $\canoEq$ is at least $n!$.
\end{proof}

\subsection{Proofs in Section~\ref{sec:improved-ramsey-for-fdfws}}
\label{app:improved-rc}

\sizeOfCanoEq*
\begin{proof}
One can uniquely encode each equivalence class $\iclass{u}$ as the set of states $\trans(\inits, u)$.
Therefore, the number of equivalence classes of $\icanoEq$ is at most $2^{n}$.
\end{proof}

\sizeOfproEq*
\begin{proof}

Similarly to Lemma~\ref{lem:size-of-canoEq-old}, we can uniquely encode an equivalence class $\proclass{v}$ as a set $S \subseteq \R = (\trans(\inits, u) \times \setnocond{0, 1} \times \states)$.
Thus, $\size{\proEq_{u}} \leq 3^{n^{2}}$.
\end{proof}

\compactnessOfImprovedRc*

\begin{proof}
Assume that $v_{1} \canoEq v_{2}$.
For each pair of states $q \in \trans(\inits, u)$ and $r \in \states$, if $q \pathto{v_{1}}{\trans} r$, we also have $q \pathto{v_{2}}{\trans} r$ by Definition~\ref{def:congruence-rel}, since $v_{1} \canoEq v_{2}$.
Analogously, $q \fpathto{v_{1}}{\trans} r$ implies that $q \fpathto{v_{2}}{\trans} r$.
Similarly, we can prove that $u_{1} \canoEq u_{2}$ implies that $u_{1} \icanoEq u_{2}$ for all $u_{1}, u_{2} \in \finwords$.
\end{proof}

\comparisonCanoEq*
\begin{proof}
The family of NBWs $\B_{n}$, inspired from~\cite{AngluinF16}, is depicted in Fig.~\ref{fig:family-nbw-ryc-smaller} with $n+3$ states.
Let $\trans_{n}$ and $\states$ be the transition function and the set of states of $\B_{n}$, respectively.
We can see from Fig.~\ref{fig:family-nbw-ryc-smaller} that the initial state is $q$ and $\acc = \setnocond{q, q_{-1}}$.

We first show that for each $u \in \finwords$, $\size{\iclass{u}} \leq (n+3) + 1$, which also indicates that the FDFW $\F$ constructed from $\B_{n}$ has $\bigO(n^{2})$ states.
By Definition~\ref{def:right-congruence}, each equivalence class $\iclass{u} $ can be encoded as $\trans(q, u)$.
Therefore, the number of equivalence classes of $\icanoEq$ is $n+4$, namely $\setnocond{q_{-1}}, \setnocond{q_{0}}, \setnocond{q_{1}}, \cdots, \setnocond{q_{n}}, \setnocond{q}, \setnocond{q_{0}, q_{-1}}$.
Next we show that for an equivalence class $\iclass{u}$, the number of equivalence classes of $\proEq_{u}$ is at most $(n+3) + 1$.

Let $v \in \poswords$.
For an equivalence class, say $\iclass{u} = \setnocond{q_{i}}$ where $q_{i} \in \states \setminus \setnocond{q_{-1}, q_{0}, q}$ for some $1 \leq i \leq n$, then $\B_{n}$ will reach a state $r$ from $q_{i}$ over $v$ with $r \in \setnocond{q, q_{0}, q_{1}, \cdots , q_{n}}$ or $\B_{n}$ will reach the set $\setnocond{q_{0}, q_{-1}}$.
We may also have $q_{i} \fpathto{v}{\trans} r$ if the path visits $q$.
Therefore, the class $\proclass{v}$ can be encoded as one of the following:
$\setnocond{q_{i} \pathto{v}{} q_{i}}$ (without visiting $q$), $\setnocond{q_{i} \pathto{v}{} q, q_{i} \fpathto{v}{} q}$, $\setnocond{q_{i} \pathto{v}{} q_{j}, q_{i} \fpathto{v}{} q_{j}}$ for all $0 \leq j \leq n, \setnocond{q_{i}\pathto{v}{}q_{-1}, q_{i} \fpathto{v}{} q_{-1}}$ and $\setnocond{q_{i} \pathto{v}{} q_{0}, q_{i} \fpathto{v}{} q_{0}, q_{i} \pathto{v}{} q_{-1}, q_{i} \fpathto{v}{} q_{-1}}$.
It follows that for the equivalence class $\iclass{u} = \setnocond{q_{i}}$ where $q_{i} \in \states \setminus \setnocond{q_{-1}, q, q_{0}}$ for some $1 \leq i \leq n$, the number of possible classes $\proclass{v}$ is at most $(n+3) + 2$.

For $\iclass{u} = \setnocond{q}$, similarly, the equivalence class $\proclass{v}$ can be encoded as one of the following:
$\setnocond{q \pathto{v}{} q, q \fpathto{v}{} q}$, $\setnocond{q \pathto{v}{} q_{j}, q \fpathto{v}{} q_{j}}$ for all $0 \leq j \leq n, \setnocond{q\pathto{v}{}q_{-1}, q \fpathto{v}{} q_{-1}}$ and $\setnocond{q \pathto{v}{} q_{0}, q \fpathto{v}{} q_{0}, q \pathto{v}{} q_{-1}, q \fpathto{v}{} q_{-1}}$.
In total, the number of possible $\proclass{v}$ is at most $n+3 \leq (n+3) + 2$.

For $\iclass{u} = \setnocond{q_{0}}$, the class $\proclass{v}$ can be encoded as either $\setnocond{q_{0} \pathto{v}{} q_{0}}$, $\setnocond{q_{0}\pathto{v}{}q_{-1}, q_{0} \fpathto{v}{} q_{-1}}$ or $\setnocond{q_{0} \pathto{v}{} q_{0}, q_{0} \pathto{v}{} q_{-1}, q_{0} \fpathto{v}{} q_{-1}}$.
So, the number of possible $\proclass{v}$ is at most $3$.
Similarly, for $\iclass{u} = \setnocond{q_{0}, q_{-1}}$ and $\setnocond{q_{-1}}$, the number of possible $\proclass{v}$ is also at most $3$.

Therefore, we have $\size{\proEq_{u}} \leq (n+3) + 2$ for each given $u \in \finwords$.

Now we show that the number of equivalence classes of $\canoEq$ in Definition~\ref{def:congruence-rel} is at least $n!$.

\emph{Short version of the proof.}
To prove that $\size{\canoEq} \geq n!$, one can just show that for each pair of different permutation words $u = i_{1}, \cdots, i_{n}$ and $u' = i'_{1}, \cdots, i'_{n}$ of $\setnocond{1, \cdots, n}$, $u \not\canoEq u'$.
We denote by $k$ the smallest position such that $i_{k} \neq i'_{k}$ where $1 \leq k < n$.
If $k = 1$, i.e., $i_{1} \neq i'_{1}$, then $u \not\canoEq u'$ by Definition~\ref{def:congruence-rel}, since we have $q \pathto{u}{} q_{i_{1}}$ and $q \pathto{u'}{} q_{i'_{1}}$ with $q_{i_{1}} \neq q_{i'_{1}}$.
Note that $\trans_{n}$ is deterministic on words without the letter $0$.
Thus $q$ transitions to $q_{i_{1}}$ after reading $i_{1}$ and stays there for the remaining $i_{2}, \cdots, i_{n}$, i.e., $q \pathto{u}{} q_{i_{1}}$.
Similarly, we have $q \pathto{u'}{} q_{i'_{1}}$.
When $1 < k < n$, we have $i_{k-1} = i'_{k-1}$.
Analogously, $u \not\canoEq u'$ since we have $q_{i_{k-1}} \pathto{u}{} q_{i_{k}}$ while $q_{i_{k-1}} \pathto{u'}{} q_{i'_{k}}$ with $i_{k} \neq i'_{k}$.
Therefore, we conclude that $u \not\canoEq u'$.
Thus, the number of equivalence classes of $\canoEq$ is at least $n!$.

\emph{Long version of the proof.}
Consider the NBW $\B_{n}$ depicted in Fig.~\ref{fig:family-nbw-ryc-smaller}:
to prove the above claim, we can just show that for each pair of different permutation words $u = i_{1}, \cdots, i_{n}$ and $u' = i'_{1}, \cdots, i'_{n}$ of $\setnocond{1, \cdots, n}$, we have that $u \not\canoEq u'$.
We first define an index function $g_{u} \colon \setnocond{1, \cdots, n} \to \setnocond{1, \cdots, n}$ for each $u = i_{1}, i_{2}, \cdots, i_{n}$ such that $g(i_{k}) = k$ for each $k \geq 1$.
That is, $g_{u}$ obtains the index $k$ of a letter $i_{k}$ in $u$.
Obviously, $g_{u} \neq g_{u'}$ if $u \neq u'$.
For instance, when $u = 1, 2, \cdots, n$, we have $g_{u}(k) = k$ for each $k \geq 1$.
Now we define another function $f_{u} \colon \setnocond{q, q_{1}, \cdots, q_{n}} \to \setnocond{q, q_{1}, \cdots, q_{n}}$ to represent the reachability between states over the word $u$.
When $u = 1, 2, \cdots, n$, we have $q_{k-1} \pathto{u}{} q_{k}$ for each $2 \leq k \leq n$, $q \pathto{u}{} q_{1}$, $q_{0} \pathto{u}{} q_{-1}$ and $q_{-1} \pathto{u}{} q_{-1}$.
Here we can just omit $q_{-1}$ and $q_{0}$ since we have only $q_{-1} \pathto{u}{} q_{-1}$ and $q_{0} \pathto{u}{} q_{-1}$ for each permutation word $u$.
Moreover, we can omit the reachability path that visits accepting states, as this does not affect the lower bound of the complexity.
Therefore, for $u = 1, \cdots, n$, we have $f_{u}(q_{k-1}) = q_{k}$ for each $2 \leq k \leq n$, $f_{u}(q_{n}) = q$ and $f_{u}(q) = q_{1}$.
In fact, according to the transition function $\trans_{n}$, the function $f_{u}$ can be defined with $g_{u}$ on the input $u = i_{1}, \cdots, i_{n}$ as follows.
\begin{itemize}
\item
    First, $f_{u}(q) = q_{i_{1}}$. From Fig.~\ref{fig:family-nbw-ryc-smaller}, it is easy to verify that from state $q$, $\B_{n}$ first goes to state $q_{i_{1}}$ after inputting $i_{1}$ and then stays there for the remaining letters.
\item
    Second, $f_{u}(q_{j}) = q_{i_{g_{u}(j) + 1}}$ if $g_{u}(j) + 1 \leq n$, otherwise $f_{u}(q_{j}) = q$, as $g_{u}(j) = n$.
    Intuitively, $\B_{n}$ will stay at $q_{j}$ until seeing the letter $j$.
    Once the letter $j$ has been read, $\B_{n}$ moves from $q_{j}$ to $q$ and then continues by reading the letter $k = i_{g_{u}(j) + 1}$ if $g_{u}(j) < n$ and moves to state $q_{k}$ and stays there.
    If $g_{u}(j) = n$, then $\B_{n}$ will stay at $q$.
\end{itemize}
Given two different permutation words $u$ and $u'$ of $\setnocond{1, \cdots, n}$, we denote by $k$ the smallest position such that $i_{k} \neq i'_{k}$ where $1 \leq k < n$.
If $k = 1$, i.e., $i_{1} \neq i'_{1}$, then $f_{u}(q) = q_{i_{1}}$ while $f_{u'}(q) = q_{i'_{1}}$.
Therefore, $u \not\canoEq u'$ since we have $q \pathto{u}{} q_{i_{1}}$ and $q \pathto{u'}{} q_{i'_{1}}$ with $q_{i_{1}} \neq q_{i'_{1}}$.
Note that $\trans_{n}$ is deterministic on words without the letter $0$.

Otherwise we have $i_{k} \neq i'_{k}$ for $1 < k < n$ and $i_{k-1} = i'_{k-1}$.
Then we have $f_{u}(q_{i_{k-1}}) = q_{i_{k}}$ while $f_{u'}(q_{i_{k-1}}) = q_{i'_{k}}$.
Then $u \not\canoEq u'$ since we have $q_{i_{k-1}} \pathto{u}{} q_{i_{k}}$ and $q_{i_{k-1}} \pathto{u'}{} q_{i'_{k}}$ with $i_{k} \neq i'_{k}$.
We then conclude that $u \not\canoEq u'$.
Thus, the number of equivalence classes of $\canoEq$ is at least $n!$, since this is the number of permutation words of $\setnocond{1, \cdots, n}$.
\end{proof}

\dbwSizeImprovedFDFW*
\begin{proof}
Let $\trans$, $s$ and $\states$ be the transition function, the initial state and the set of states of $\A$, respectively.
Since an equivalence class $\iclass{u}$ can be uniquely encoded as $\trans(s, u)$, the number of equivalence classes of $\icanoEq$ is $n$ since $\trans$ is deterministic.
For an equivalence class $\iclass{u}$, an equivalence class $\proclass{v}$ can be uniquely encoded as a set $S$ of the reachability between states in $\setnocond{q} $ and in $\states$ where $q = \trans(s, u)$.
Since $\trans$ is deterministic, we have that $\size{\trans(q, v)} = 1$ and there are at most $2$ pairs of states in $S$.
This is because there is at most one state $r$ in $q \pathto{v}{\trans} r$ once $v$ is given.
Without loss of generality, for an equivalence class $\proclass{v}$, assume that $r = \trans(q, v)$.
Then, we have that $S=\setnocond{q\pathto{v}{} r}$ or $S=\setnocond{q\pathto{v}{} r, q\fpathto{v}{} r}$.
It follows that the number of possible $S$ for the equivalence class $\proclass{v}$ is at most $2$.
Therefore, the number of equivalence classes of $\proEq_{u}$ is $2n$, as the number of possible $q$ is $n$.
Then $\Sigma_{\iclass{u} \in \iquotient} \size{\proEq_{u}} \leq n \times 2n \in \bigO(n^{2})$.
\end{proof}

\saturationLemmaImproved*
\begin{proof}
Consider Item~(1) and two words $u \in \finwords, v \in \poswords$ such that $uv \icanoEq u$.
By Definition~\ref{def:right-congruence}, we have that $\trans(\inits, u) = \trans(\inits, uv)$.
For each word $v' \in \proclass{v}$, states $q \in \trans(\inits, u)$ and $r \in \states$, by Definition~\ref{def:pro-right-congruence} we have that $q \pathto{v}{} r$ implies $q \pathto{v'}{} r$.
Thus, $\trans(\inits, uv') = \trans(\inits, uv) = \trans(\inits, u)$, hence $\iclass{u}\proclass{v} = \iclass{u}$.

If $\iclass{u}\proclass{v}^{\omega} \cap \lang{\A} = \emptyset$, then Item~(1) trivially holds.
Now assume that there exists a word $w \in \iclass{u}\proclass{v}^{\omega} \cap \lang{\A}$.
We can then write $w$ as $w = u_{0} \cdot v_{1} \cdot v_{2} \cdots $ with $u_{0} \in \iclass{u}$ and $v_{i} \in \proclass{v}$ for each $i \geq 1$.
It follows that $\trans(\inits, u_{0}) = \trans(\inits, u_{0}v_{i}) = \trans(\inits, u_{0}v_{1} \cdots v_{i})$ for each $i \geq 1$ since $u_{0}v_{i} \icanoEq u_{0}$ for each $i \geq 1$.
Therefore, the set of states $\trans(\inits, u_{0}) = \trans(\inits, u)$ has been visited infinitely many times and they can be reached by themselves.
Since $w \in \lang{\A}$, there exists an accepting run of $\A$ over $w$ that can be written as $\rho_{w} = q \pathto{u_{0}}{\trans} q_{0} \pathto{v_{1}}{\trans} q_{1} \pathto{v_{2}}{\trans} q_{2} \cdots$ where $q \in \inits, q_{0} \in \trans(q, u_{0})$ and $q_{i} \in \trans(\inits, u_{0}v_{i})$ for each $i\geq 0$.
By Definition~\ref{def:right-congruence}, there is the accepting run $\rho = p \pathto{u}{\trans} q_{0} \pathto{v_{1}}{\trans} q_{1} \pathto{v_{2}}{\trans} q_{2} \cdots$ for some $p \in \inits$ such that $p \pathto{u}{} q_{0}$, since $q_{0} \in \trans(\inits, u)$.
Let $w' = u'_{0} v'_{1} v'_{2} \cdots$ with $ u'_{0} \in \iclass{u}$ and $v'_{i} \in \proclass{v}$ for each $i\geq 1$.
By Definition~\ref{def:pro-right-congruence}, $\rho$ is also an accepting run of $\A$ over $u \cdot v'_{1} \cdot v'_{2}\cdots$, since $q_{i-1} \pathto{v_{i}}{} q_{i}$, as $v_{i} \proEq_{u} v'_{i}$ for each $k \geq 1$ and $q_{i} \in \trans(\inits, u)$ for each $k \geq 0$.
Since $u'_{0} \in \iclass{u}$, there exists a state $q' \in \inits$ such that $q' \pathto{u'_{0}}{} q_{0}$.
It follows that $\rho_{w'} = q' q_{0} q_{1} \cdots$ is a run of $\A$ over $w'$.
Clearly, $\rho_{w'}$ is an accepting run of $\A$.
Therefore, $w'$ is also accepted by $\A$.
It follows that $\iclass{u}\proclass{v}^{\omega} \subseteq \lang{\A}$ if $\iclass{u}\proclass{v}^{\omega} \cap \lang{\A} \neq \emptyset$.

Regarding Item~(2), we are going to show the equality of UP-words for proving the equivalence of the two languages, by means of Theorem~\ref{thm:upword-omega-regular-language}.
Let $w \in \upword{\infwords}$, $(u', v')$ be a decomposition of $w$ and $\M$ be the DFW induced by $\icanoEq$ as by Definition~\ref{def:induced-dfw}.
According to~\cite{AngluinBF18}, there exist $u = u'v'^{h}$ and $v = v'^{k}$, where $h, k \geq 1$ such that $\M(u) = \M(uv)$, i.e., $(u, v)$ is a normalized decomposition of $w$.
Intuitively, since $\M$ is deterministic and has finitely many states, $\M$ visits a state, say $\M(u)$, from $\M(u')$ twice after inputting enough repetitions of $v'$.
Without loss of generality, let $\iclass{u} = \M(u)$.
Since $\M(u) = \M(uv)$, we have that $u \icanoEq uv$.
It follows that we have an $\omega$-language $\iclass{u}\proclass{v}^{\omega}$ with $uv\icanoEq u$ such that $ w \in \iclass{u}\proclass{v}^{\omega}$ for each $\omega$-word $w \in \upword{\infwords}$.
Therefore $\infwords = \bigcup \setcond{\iclass{u}\proclass{v}^{\omega}}{u \in \finwords, v \in \poswords, uv\icanoEq u}$.

Item~(3) is a direct consequence of Items~(1) and (2).
By Item~(2), the decompositions $(u, v)$ with $uv \icanoEq u$ covers the whole set of infinite words;
moreover, according to Item~(1), for each of such decompositions, either $\iclass{u}\proclass{v}^{\omega} \cap \lang{\A} = \emptyset$ or $\iclass{u}\proclass{v}^{\omega} \subseteq \lang{\A}  $.
Therefore, to obtain $\infwords \setminus \lang{\A}$, we only need to take the union of all the $\omega$-regular languages $\iclass{u}\proclass{v}^{\omega}$ such that $\iclass{u}\proclass{v}^{\omega} \cap \lang{\A} = \emptyset$ and $uv \icanoEq u$.
\end{proof}

\subsection{Proofs in Section~\ref{sec:optimal-rc}}
\label{app:opt-rc}

\subsetOfMaxPreds*

\begin{proof}
Let $k = \size{u}$ and $q' \in [q]_{\preceq_{uv}}$.
We first prove that for a run $\pi \in \max(\Pi_{uv, q'})$, we have $\pi[k+1] \simeq_{u} p$.
Since both $\pi[k+1]$ and $p$ belong to $\trans(\inits, u)$, we know that $\pi[k+1]$ and $p$ are comparable under $\preceq_{u}$ according to Definition~\ref{def:preorder-u}.

For each run $\pi \in \max(\Pi_{uv, q'})$, we first prove by contradiction that $p \preceq_{u} \pi[k+1]$.
If $\pi[k+1] \prec_{u} p$, i.e., $p \not\preceq_{u} \pi[k+1] $, then $\pi$ is not in $\max(\Pi_{uv, q'})$, since there exists a run $\pi'$ greater than $\pi$ that can be obtained by extending $\pi_{p} \in \max(\Pi_{u,p})$ with $p \pathto{v}{} q$.
Since $\pi'$ reaches state $q$, it follows that $\pi'$ is greater than $\pi$, which means that $q' \prec_{uv} q$ according to Definition~\ref{def:preorder-u}.
This contradicts $q' \in [q]_{\preceq_{uv}}$.
Thus we have that $p \preceq_{u} \pi[k+1]$ for all $\pi \in \max(\Pi_{uv, q'})$.

Analogously, we can prove that $\pi[k+1] \preceq_{u} p$.
If $p \prec_{u} \pi[k+1]$, i.e., $\pi[k+1] \not\preceq_{u} p$, then $[p]_{\preceq_{u}}$ is not the maximal equivalence class that has a state $p'$ reaching $q$ over $v$.
Therefore, we have $\pi[k+1] \simeq_{u} p$ for each run $\pi \in \max(\Pi_{uv,q'})$.
Since $q' \in [q]_{\preceq_{uv}}$ has been taken arbitrarily, it follows that for each $q' \in [q]_{\preceq_{uv}}$, $\setcond{\pi[\size{u}+1]}{\pi \in \max(\Pi_{uv,q'})} \subseteq [p]_{\preceq_{u}}$.
\end{proof}

\saturationAccPreorder*
\begin{proof}
Recall that $\max(\Pi_{uv, q})|_{u} = \setcond{\pi[\size{u} + 1]}{\pi \in \max(\Pi_{uv, q})}$.

We prove the claim by contradiction.
Assume without loss of generality that $p_{1} \fpathto{v}{} q_{1}$ but not $p_{2} \fpathto{v}{} q_{2}$.
It is obvious that $p_{1} \pathto{v}{} q_{1}$ and $p_{2} \pathto{v}{}q_{2}$ since $p_{1}$ and $p_{2}$ are on the maximal runs to $q_{1}$ and $q_{2}$, respectively.
Let $\pi_{q_{1}} \in \max(\Pi_{uv, q_{1}})$ with $\pi_{q_{1}}[\size{u} + 1] = p_{1}$ and $\pi_{q_{2}} \in \max(\Pi_{uv,q_{2}})$ with $\pi_{q_{2}}[\size{u} + 1] = p_{2}$.
Since, $q_{1} \simeq_{uv} q_{2}$, we know that $\pi_{q_{1}}[i] \in \acc \Longleftrightarrow \pi_{q_{2}}[i]\in \acc$ for all $1 \leq i \leq \size{uv} +1$.
Since $p_{2} \fpathto{v}{} q_{2}$ does not hold, we have $\pi_{q_{2}}[i] \notin \acc$ as well as $\pi_{q_{1}}[i] \notin \acc$ for all $\size{u} < i \leq \size{uv} + 1$.
However, since $p_{1} \fpathto{v}{} q_{1}$, i.e., $p_{1} \pathto{v}{} q_{1}$ and it visits an accepting state, there exists a run $\pi'_{q_{1}}$ greater than $\pi_{q_{1}}$ that can be obtained by extending $\pi_{q_{1}}[1\cdots \size{u}]$ with $p_{1}\fpathto{v}{} q_{1}$, because of the visit to an accepting state between $p_{1}$ and $q_{1}$.
(If $\size{u} = 0$, then $\pi_{q_{1}}[1\cdots \size{u}]$ is empty sequence.)
This contradicts the fact that $\pi_{q_{1}} \in \max(\Pi_{uv,q_{1}})$, by which the claim follows.
\end{proof}

\sizeOfOptCanoEq*
\begin{proof}
For $u \in \finwords$, an equivalence class $\oclass{u}$ can be uniquely encoded as an ordered sequence of sets $\langle S_{1}, S_{2}, \cdots, S_{k} \rangle$, i.e,. $\trans(\inits, u)/_{\preceq_{u}}$ where $\trans(\inits, u) = \cup_{1\leq j\leq k} S_{j}$ and $S_{i}\cap S_{j} = \emptyset$ for each $1 \leq i < j \leq k$ if $k > 1$.
Here we have that $S_{i} \classpre S_{j}$ for each $1 \leq i < j \leq k$ if $k > 1$.
According to~\cite{Fogarty13}, the number of possible ordered sequence of sets over $\states$ is approximately $(0.53n)^{n}$.
Thus we have that $\size{\canoEq^{o}} \leq n^{n}$.
\end{proof}

\sizeOfOptProEq*
\begin{proof}
Given $u$ is fixed, let $v \in \finwords$.
Each equivalence class $[v]_{\proEq^{o}_{u}}$ can be uniquely encoded as a pair $\langle \trans(\inits, uv)/_{\preceq_{uv}}, f \rangle$ where
\begin{itemize}
\item
    $\trans(\inits, uv)/_{\preceq_{uv}}$ is the set of equivalence classes of $\trans(\inits, uv)$ under the preorder $\preceq_{uv}$;
\item
    $f$ is a function mapping a state $q \in \trans(\inits, uv)$ to the maximal equivalence class in $\trans(\inits, u)/_{\preceq_{u}}$ that has a state $p$ reaching $q$ over $v$ (corresponding to (2)-(i) of Definition~\ref{def:opt-proEq}) and a Boolean value $\ell$ that marks whether there is a maximal run $\pi \in \max(\Pi_{uv, q})$ with $\pi[\size{u}+1] \fpathto{v}{} q$  (corresponding to (2)-(ii) of Definition~\ref{def:opt-proEq}) or an empty set that indicates the state $q$ is not present in $\trans(\inits, uv)$.
\end{itemize}
Given two words $v_{1}, v_{2} \in \finwords$, we represent $\oproclass{v_{1}}$ with $\langle \trans(\inits, uv_{1}), f_{1} \rangle$ and $\oproclass{v_{2}}$ with $\langle \trans(\inits, uv_{2}), f_{2} \rangle$.
According to Definition~\ref{def:opt-proEq}, it is easy to see that $v_{1} \proEq^{o}_{u} v_{2}$ iff $\langle \trans(\inits, uv_{1}), f_{1} \rangle = \langle \trans(\inits, uv_{2}), f_{2} \rangle$.

Since $\trans(\inits, u)/_{\preceq_{u}}$ is given and there are at most $n$ equivalence classes in $\trans(\inits, u)/_{\preceq_{u}}$, the number of possible functions $f$ is at most $(2(n+1))^{n} = (n+1)^{n} \times 2^{n}$,
where the factor $2$ comes from the Boolean $\ell$ and the factor $n+1$ from the empty set and the equivalence classes.
It is then clear that there are at most $n^{n} \times (n+1)^{n} \times 2^{n}$ equivalence classes defined by $\proEq^{o}_{u}$ since, by~\cite{Fogarty13}, the number of possible $\trans(\inits, uv)/_{\preceq_{uv}}$ is at most $n^{n}$ regardless of what the word $uv$ is.
\end{proof}

\dbwSizeOptFDFW*
\begin{proof}
Let $\trans$, $s$ and $\states$ be the transition function, the initial state and the set of states of $\A$, respectively.
Recall that an equivalence class $\oclass{u}$ can be uniquely encoded as $\trans(s, u)/_{\preceq_{u}}$.
Since $\trans$ is deterministic, i.e., $\size{\trans(s, u) } = 1$, so there is only one equivalence class in $\trans(s, u)/_{\preceq_{u}}$, i.e., $\setnocond{\trans(s, u)}$.
So the number of equivalence classes of $\canoEq^{o}$ is at most $n$ since each $\setnocond{\trans(s, u)}$ is a singleton.

When we fix a word $u$, we can just represent $u$ as the state $\trans(s, u)$, say $q$.
As mentioned in the proof of Lemma~\ref{lem:upper-bound-opt-proEq}, an equivalence class $\oproclass{v}$ can be uniquely encoded as $\langle \trans(\inits, uv)/_{\preceq_{uv}}, f\rangle$.
Again, since $\trans$ is deterministic, $\trans(\inits, uv)/_{\preceq_{uv}}$ is a singleton, so the number of possible singletons is at most $n$.
Let $r = \trans(\inits, uv) \in \states$.
Recall that $f$ is a function mapping a state $q \in \trans(\inits, uv)$ to the maximal equivalence class in $\trans(\inits, u)/_{\preceq_{u}}$ that has a state $p$ reaching $q$ over $v$ (corresponding to (2)-(i) of Definition~\ref{def:opt-proEq}) and a Boolean value $\ell$ that marks whether there is a maximal run $\pi \in \max(\Pi_{uv, q})$ with $\pi[\size{u}+1] \fpathto{v}{} q$  (corresponding to (2)-(ii) of Definition~\ref{def:opt-proEq}) or an empty set that indicates the state $q$ is not present in $\trans(\inits, uv)$.
Since $\trans$ is deterministic, the maximal equivalence class in $\trans(\inits, u)/_{\preceq}$ is a constant since $u$ is fixed and $\trans(\inits, u)/_{\preceq}$ has to be $\setnocond{q}$.
For the Boolean value $\ell$, there are $2$ possible values.
Thus the number of possible $\oproclass{v}$ is at most $2 \times n$.
Since $\size{\canoEq^{o}} \leq n$, then the claim follows.
\end{proof}

In order to prove Lemma~\ref{lem:saturation-opt-fdfws}, below we show that for each $w \in \oclass{u} \oproclass{v}^{\omega}$, we can find an accepting run $\rho$ whose prefixes $\rho[1\cdots i+1]$ is a maximal run in $\Pi_{w[1\cdots i], \rho[i+1]}$ for all $i \geq 1$.
\begin{restatable}{lemma}{maxRunForACC}
\label{lem:existence-max-run}
Given $u \in \finwords, v \in \poswords$ with $uv \canoEq^{o} u$, let $w = u_{0}v_{1} v_{2} \cdots$ where $u_{0} \in [u]_{\canoEq^{o}}$ and $v_{j} \in \oproclass{v}$ for all $j \geq 1$.
If $w \in \lang{\A}$, then there exists an accepting run $\rho = q \pathto{u_{0}}{} q_{0} \pathto{v_{1}}{} q_{1} \pathto{v_{2}}{}\cdots$ of $\A$ over $w$ where $q \in \inits$ and $[q_{j}]_{\preceq_{u}} = \max_{\preceq_{u}} \setcond{[p]_{\preceq_{u}} \in \trans(\inits, u)/_{\preceq_{u}}}{p \pathto{v_{j+1}}{} q_{j+1}}$ for each $j \geq 0$.
\end{restatable}


\begin{proof}
Among all accepting runs of $\A$ over $w$, we can find a run $\rho$ such that for all $i \geq 1$, there does not exist another accepting run $\rho'$ with $\rho'[1\cdots i]$ being greater than $\rho[1\cdots i]$.
It is easy to see that such an accepting run $\rho$ exists, which we call a \emph{maximal} accepting run.

Since $uv \canoEq^{o} u$, similarly to the proof of Lemma~\ref{lem:saturation-impoved-fdfws}, $\oclass{u}\oproclass{v} = \oclass{u}$.
It follows that $u \canoEq^{o} u_{0}v_{i} \canoEq^{o} u_{0}v_{1}\cdots v_{i}$ for each $i \geq 1$ since $u_{0}v_{i} \canoEq^{o} u_{0} \canoEq^{o} u$ for each $i \geq 1$.
By Definition~\ref{def:opt-canoEq}, we have $\trans(\inits, u)/_{\preceq_{u}} = \trans(\inits, uv_{i})/_{\preceq_{uv_{i}}} = \trans(\inits, uv_{1} \cdots v_{i})/_{\preceq_{uv_{1} \cdots v_{i}}}$ for every $i \geq 1$ since $u \canoEq^{o} u_{0}v_{i} \canoEq^{o} u_{0}v_{1}\cdots v_{i}$ for each $i \geq 1$.
We also have $\trans(\inits, u) = \trans(\inits, uv_{i}) = \trans(\inits, uv_{1} \cdots v_{i})$ for every $i \geq 1$.
Let $P = \trans(\inits, u) = \trans(\inits, uv_{1} \cdots v_{i})$ where $i \geq 1$.

We write such a maximal accepting run $\rho$ as $\rho = q \pathto{u_{0}}{} q_{0} \pathto{v_{1}}{} q_{1} \pathto{v_{2}}{} \cdots$ where $q \in \inits$ and $q_{j} \in P$ for all $j \geq 0$.
In order to simplify the notation, let $v_{0} = \emptyword$ and $[p]_{\preceq_{u_{0} v_{0} \cdots v_{j}}} = \max_{\preceq_{u_{0} v_{0} \cdots v_{j}}}\setcond{[p']_{\preceq_{u_{0} v_{0} \cdots v_{j}}} \in P/_{\preceq_{u_{0} v_{0} \cdots v_{j}}}}{p' \pathto{v_{j+1}}{} q_{j+1}}$ for all $j \geq 0$.
Since $p, q_{j} \in \trans(\inits, u_{0} v_{0} \cdots v_{j}) = P$, we know that $q_{j}$ and $p$ are comparable under $\preceq_{u_{0} v_{0} \cdots v_{j}}$ according to Definition~\ref{def:preorder-u}.
Below we prove that $q_{j } \simeq_{u_{0} v_{0} \cdots v_{j} }p$, i.e., $[q_{j}]_{\preceq_{u_{0} v_{0} \cdots v_{j}}} \simeq_{u_{0} v_{0} \cdots v_{j}} [p]_{\preceq_{u_{0} v_{0} \cdots v_{j}}}$.
By definition, there exists a state $p' \in [p]_{\preceq_{u_{0} v_{0} \cdots v_{j}}}$ such that $p' \pathto{v_{j+1}}{} q_{j+1}$.
Then we only need to prove that $p' \simeq_{u_{0} v_{0} \cdots v_{j}} q_{j}$.

We prove by contradiction that $p' \preceq_{u_{0} v_{0} \cdots v_{j}} q_{j}$ holds, so assume that $q_{j} \prec_{u_{0} v_{0} \cdots v_{j}} p'$, i.e., $p' \not\preceq_{u_{0} v_{0} \cdots v_{j}} q_{j} $.
This implies that $\rho$ cannot be a maximal accepting run over $w$, since we can obtain an accepting run $\rho'$ by extending $q' \pathto{u_{0} v_{0} \cdots v_{j}}{} p' \pathto{v_{j+1}}{} q_{j+1}$ for some $q' \in \inits$ with $q_{j+1} \pathto{v_{j+2}}{} q_{j+2} \cdots $ such that the prefix $q' \pathto{u_{0} v_{0} \cdots v_{j}}{} p'$ of $\rho'$ is greater than the prefix $q \pathto{u_{0} v_{0} \cdots v_{j}}{} q_{j}$ according to Definition~\ref{def:preorder-u};
here we have $q' \pathto{u_{0} v_{0} \cdots v_{j}}{}p'\in \max(\Pi_{u_{0} v_{0} \cdots v_{j}, p'})$.
Note that according to Corollary~\ref{coro:preceq-init-to-q}, if $q \fpathto{u_{0} v_{0} \cdots v_{j}}{}q_{j}$, then we can always find the above initial state $q'$ such that $q' \fpathto{u_{0} v_{0} \cdots v_{j}}{} p'$.
Thus the constructed run $\rho'$ will be accepting even if $\rho$ eventually reaches an infinite repetition of a set of states.
Since $\rho'$ is greater than $\rho$, this contradicts the fact that $\rho$ is maximal.
Therefore we have that $p' \preceq_{u_{0} v_{0} \cdots v_{j}} q_{j}$.

Analogously, we can prove that $q_{j}\preceq_{u_{0} v_{0} \cdots v_{j}} p'$.
If $p' \prec_{u_{0} v_{0} \cdots v_{j}} q_{j}$, i.e., $q_{j} \not\preceq_{u_{0} v_{0} \cdots v_{j}} p'$, then $[p']_{\preceq_{u_{0} v_{0} \cdots v_{j}}}$ is not the maximal equivalence class that has a state $p'$ reaching $q_{j+1}$ over $v$.
Therefore, we have $q_{j} \simeq_{u_{0} v_{0} \cdots v_{j}} p'$.
Since we have proved that $\trans(\inits, u)/_{\preceq_{u}} = \trans(\inits, uv_{j})/_{\preceq_{u_{0} v_{0} \cdots v_{j}}}$ for all $j \geq 0$, it follows that $[q_{j}]_{\preceq_{u}} = \max_{\preceq_{u}}\setcond{[p]_{\preceq_{u}} \in P/_{\preceq_{u}}}{p \pathto{v_{j+1}}{}q_{j+1}}$ for all $j \geq 0$, as required.
\end{proof}

\optSaturationLemma*
\begin{proof}
Here we only prove Item~(1).
The proofs for Items~(2) and (3) are minor adaptations of the corresponding ones for Lemma~\ref{lem:saturation-impoved-fdfws}.

If $\oclass{u} \oproclass{v}^{\omega} \cap \lang{\A} = \emptyset$, then Item~(1) trivially holds, so assume that $\oclass{u} \oproclass{v}^{\omega} \cap \lang{\A} \neq \emptyset$;
we shows that this implies that $\oclass{u}\oproclass{v}^{\omega} \subseteq \lang{\A}$.
Let $P = \trans(\inits, u)$ and $w \in \oclass{u}\oproclass{v}^{\omega} \cap \lang{\A} \neq \emptyset$;
this means that $w$ can be written as $w = u_{0} \cdot v_{1} \cdot v_{2} \cdots $ with $u_{0} \in \oclass{u}$ and $v_{i} \in \oproclass{v}$ for each $i \geq 1$.
Since $w \in \lang{\A}$, according to Lemma~\ref{lem:existence-max-run}, there exists an accepting run of $\A$ over $w$ that can be written as $\rho_{w} = q \pathto{u_{0}}{\trans} q_{0} \pathto{v_{1}}{\trans} q_{1} \pathto{v_{2}}{\trans} q_{2} \cdots$ where $q \in \inits$ and  $[q_{j}]_{\preceq_{u}} = \max_{\preceq_{u}}\setcond{[p]_{\preceq_{u}} \in P/_{\preceq_{u}}}{p \pathto{v_{j+1}}{} q_{j+1}}$ for each $j \geq 0$.

Let $w'$ be an arbitrary word in $\oclass{u} \oproclass{v}^{\omega}$, i.e., $w' = u'_{0} v'_{1} v'_{2} \cdots$ with $u'_{0} \in \oclass{u}$ and $v'_{i} \in \oproclass{v}$ for each $i\geq 1$.
We now show how to construct an accepting run over $w'$, thus we prove that $w' \in \lang{\A}$, i.e., $\oclass{u} \oproclass{v}^{\omega} \subseteq \lang{\A}$, as required, by proving that the constructed accepting run $\rho_{w'}$ is such that $\rho_{w'} = q' \pathto{u'_{0}}{} q'_{0} \pathto{v'_{1}}{}q'_{1} \cdots$ where $q' \in \inits$ and $q'_{j} \simeq_{u} q_{j}$ for each $j \geq 0$.
Note that for each state $q_{j+1}$ of $\rho_{w}$, since $v_{j+1} \proEq^{o}_{u} v'_{j+1} \proEq^{o}_{u} v$, we have that $[q_{j}]_{\preceq_{u}} = \max_{\preceq_{u}}\setcond{[p]_{\preceq_{u}} \in P/_{\preceq_{u}}}{p \pathto{v'_{j+1}}{} q_{j+1}} $ for all $j \geq 0$, by Item~(2)-(i) of Definition~\ref{def:opt-proEq}.
If $q_{j} \pathto{v'_{j+1}}{} q_{j+1}$ holds, then we set $q'_{j} = q_{j}$;
otherwise, according to Lemma~\ref{lem:subseteq-class-u}, for each $q'_{j+1} \in [q_{j+1}]_{\preceq_{uv_{0}\cdots v_{j+1}}} = [q_{j+1}]_{\preceq_{u}}$, there is a state $q'_{j} \in [q_{j}]_{\preceq_{u}}$ such that $q'_{j} \pathto{v'_{j+1}}{} q'_{j+1}$.
So there exists a state $q'_{j} \in [q_{j}]_{\preceq_{u}}$ such that $q'_{j} \pathto{v_{j+1}}{} q_{j+1}$.
No matter how large $j$ is, by moving backward we can always find a state $q'_{j-1} \in [q_{j-1}]_{\preceq_{u}}$ that can reach $q'_{j}$ over $v'_{j}$.
Since the number of equivalence classes in $P/_{\preceq_{u}}$ and $\states$ are both bounded by $n = \size{\states}$, we can find each state $q'_{j} \in [q_{j}]_{\preceq_{u}}$ for which $q'_{j} \pathto{v'_{j+1}}{} q'_{j+1}$ holds, for all $j \geq 0$.
In this way, we are able to construct a run $\rho_{w'} = q' \pathto{u'_{0}}{} q'_{0} \pathto{v'_{1}}{} q'_{1} \cdots$ where $q' \in \inits$ and $q'_{j} \simeq_{u} q_{j}$ for every $j \geq 0$.
We note that when choosing an appropriate $q'_{j}$, we can always choose a state $q'_{j}$ which is on a run $\pi \in \max(\Pi_{u'_{0}v'_{j+1}}, q'_{j+1}) = \max(\Pi_{uv'_{j+1}}, q'_{j+1})$ given a state $q'_{j+1}$, i.e., we choose $q'_{j} = \pi[\size{u'_{0}v'_{0} \cdots v'_{j}} + 1]$ according to Lemma~\ref{lem:subseteq-class-u}.
If $q_{j} \fpathto{v_{j+1}}{} q_{j+1}$ for $j\geq 0$, then for every run $\pi \in \max(\Pi_{uv_{j+1}}, q_{j+1})$, we have that $\pi[\size{u} + 1] \fpathto{v_{j+1}}{} q_{j+1}$ since $\pi[\size{u}+1] \simeq_{u} q_{j}$ and $q_{j} \in \max_{\preceq_{u}}\setcond{[p]_{\preceq_{u}} \in P/_{\preceq_{u}}}{p \pathto{v_{j+1}}{} q_{j+1}}$ by Lemma~\ref{lem:subseteq-class-u};
otherwise $\pi$ would not be in $\max(\Pi_{uv_{j+1}}, q_{j+1})$.
So we have $p_{1} \fpathto{v_{j+1}}{}q_{j+1}$ with $p_{1} = \pi[\size{u}+1]$.
By (2)-(ii) of Definition~\ref{def:opt-proEq}, there exists a run $\pi \in \max(\Pi_{uv'_{j+1}, q_{j+1}})$ such that $\pi[\size{u}+1] \fpathto{v'_{j+1}}{} q_{j+1}$  since $v'_{j+1} \proEq^{o}_{u}v_{j+1}$.
According to Lemma~\ref{lem:all-or-none-acc-preorder}, since $q'_{j} \in \setcond{\pi[\size{u}+1]}{\pi  \max(\Pi_{uv'_{j+1}}, q'_{j+1})}$ and $q'_{j+1} \simeq_{uv'_{j+1}} q_{j+1}$, we also have $q'_{j} \fpathto{v'_{j+1}}{} q'_{j+1}$.

Since there are infinitely many $j$s such that $q_{j} \fpathto{v_{j+1}}{} q_{j+1}$, we also have $q'_{i} \fpathto{v'_{j+1}}{} q'_{j+1}$ for infinitely many $j$s.
Therefore $\rho'_{w'}$ also visits accepting states infinitely often, thus $w' \in \lang{\A}$.
Given the arbitrary choice of $w' \in \oclass{u}\oproclass{v}^{\omega}$, it follows that $\oclass{u}\oproclass{v}^{\omega} \subseteq \lang{\A}$, as required.
\end{proof}

\subsection{Proofs of Section~\ref{sec:applications}}
\label{app:fdfw}
\correctnessOfImprovedFDFW*
\begin{proof}

To prove Item~(1), assume that $w \in \upword{\F}$.
Then there exists a normalized decomposition $(u', v')$ of $w$ such that $(u', v')$ is accepted by $\F$.
Let $u' \in \oclass{u}$ and $v' \in \oproclass{v}$.
It follows that we have $\oclass{u} = \M(u) =\M(u')= \M(uv) = \M(u'v')$ and $\oproclass{v} = \N_{u}(v')$ is an accepting macrostate.
Thus $u'v'^{\omega} \notin \lang{\A}$ by Definition~\ref{def:improved-ramsey-fdfws}, i.e., $\upword{\F} \subseteq \upword{\infwords \setminus \lang{\A}}$.
For the other direction, we assume that $w \in \upword{\infwords \setminus \inflang{\A}}$.
By Lemma~\ref{lem:saturation-opt-fdfws}, there exists an $\omega$-language $\oclass{u}\oproclass{v}^{\omega} \cap \lang{\A}$ with $uv\canoEq^{o} u$ such that $w \in \oclass{u}\oproclass{v}^{\omega}$.
Without loss of generality, let $w = uv^{\omega}$.
It follows that $\N_{u}(v)$ is an accepting macrostate, i.e., $w$ is accepted by $\F$.
Therefore, $\upword{\infwords \setminus \inflang{\A}} \subseteq \upword{\F}$.
We then have that $\upword{\F} = \upword{\infwords \setminus \inflang{\A}}$.

Consider now Item~(2).
Given two normalized decompositions $(u_{1}, v_{1})$ and $(u_{2}, v_{2})$ of $w$, assume that $(u_{1}, v_{1})$ is accepted by $\F$.
It follows that $w = u_{1}v_{1}^{\omega} = u_{2}v^{\omega}_{2}$ is not accepted by $\A$.
By Definition~\ref{def:improved-ramsey-fdfws}, let $\oclass{u} = \M(u_{2})$ and $\oproclass{v} = \N_{u}(v_2)$ where $\N_{u}$ is the progress DFW of $\oclass{u}$ of $\M$.  It follows that $u_{2} \in \oclass{u}$ and $v_{2} \in \oproclass{v}$.
According to the proof of Lemma~\ref{lem:saturation-opt-fdfws}, we have that $uv \canoEq^{o} u$ if $u_{2} v_{2}\canoEq^{o} u_{2}$.
By Lemma~\ref{lem:saturation-opt-fdfws}, we have $\oclass{u}\oproclass{v}^{\omega} \cap \lang{\A} = \emptyset$ since $u_{2}v^{\omega}_{2} \in \oclass{u}\oproclass{v}^{\omega}$ and $uv \canoEq^{o} u$.
According to Definition~\ref{def:improved-ramsey-fdfws}, $\oclass{v}$ is thus an accepting macrostate of $\N_{u}$.
It follows that $(u_{2}, v_{2})$ is also accepted by $\F$.
Thus $\F$ is saturated.

Lastly, for proving Item~(3), by Lemma~\ref{lem:size-opt-canoEq} there are at most $n^{n}$ macrostates in the leading DFW $\M$.
For each progress DFW $\N_{u}$ of $\oclass{u}$, by Lemma~\ref{lem:upper-bound-opt-proEq} there are at most $n^{n}\times (n+1)^{n} \times 2^n$ macrostates, since the number of equivalence classes of $\proEq^{o}_{u}$ is at most $n^{n}\times (n+1)^{n} \times 2^n$.
Therefore, there are $n^{n} + n^{n} \times n^{n} \times (n+1)^{n}\times 2^n \in 2^{\bigO(n \log n)}$ macrostates in $\F$.

\end{proof}

We formalize the results inspired from \cite{Yan/08/lowerComplexity} as below.
\begin{restatable}[]{theorem}{lowerBoundOfImprovedFDFW}
\label{thm:lower-bound-fdfws}
There exists a family of NBWs $\A_{1}, \cdots, \A_{n}$ with $n$ states for which an FDFW $\F_{n}$ accepting $\infwords \setminus \lang{\A_{n}}$ has $2^{\Omega(n \log n)}$ macrostates.
\end{restatable}
\begin{proof}

It is shown in~\cite{Yan/08/lowerComplexity} that there exists a family of $\omega$-languages $L_{1}, \cdots, L_{n}$ with $n \geq 1$, such that for each $n$, there exists an NBW $\A_{n}$ with $n$ states whose language is $L_{n}$, while all NBWs accepting the complementary language $L^{c}_{n}$ have $2^{\Omega(n\log n)}$ states.

Assume that the FDFW $\F_{n}$ accepting $\infwords\setminus \lang{\A_{n}}$ has $m$ macrostates such that $m < 2^{\Omega(n\log n)}$.
According to Lemma~\ref{lem:fdfw-to-nbw}, we can construct a complementary NBW $\A^{c}_{n}$ of $\A_{n}$ with $\bigO(m^3)$ states from $\F$ such that $\upword{\F} = \upword{\infwords \setminus \lang{\A_{n}}}= \upword{\lang{\A^{c}_{n}}}$.
It follows that $\A^{c}_{n}$ has $\bigO(m^3) < 2^{\Omega(n \log n)}$ states, which contradicts the results in~\cite{Yan/08/lowerComplexity}.
Therefore, $\F$ has $2^{\Omega(n \log n)}$ macrostates.
\end{proof}

\lowerBoundOptRc*
\begin{proof}
In order to prove the statement, we use results presented in Sect.~\ref{sec:applications}.
According to Definition~\ref{def:improved-ramsey-fdfws} and Definition~\ref{def:induced-dfw}, there is a bijection between a state in the FDFW $\F$ and each equivalence class defined with $(\canoEq^{o}, \bigcup_{u \in \finwords} \setnocond{\proEq^{o}_{u}})$.
That is, we have a linear translation from the right congruence relations $(\canoEq^{o}, \bigcup_{u \in \finwords} \setnocond{\proEq^{o}_{u}})$ to an FDFW $\F$ accepting $\infwords\setminus \lang{\A}$.

We prove the claim of the theorem by contradiction:
assume that there exists a family of right congruence relations $(\canoEq, \bigcup_{u \in \finwords} \setnocond{\proEq_{u}})$ satisfying the saturation condition $\class{u}\proclass{v}^{\omega} \cap \lang{\A} = \emptyset$ or $\class{u}\proclass{v}^{\omega} \subseteq \lang{\A}$ such that its number of equivalence classes is less than $2^{n \log n}$.
Since the number of equivalence classes of $(\canoEq, \bigcup_{u \in \finwords} \setnocond{\proEq_{u}})$ is bounded (less than $2^{n\log n}$) and for each $u\in\finwords, v \in \poswords$, if $uv \canoEq u$, then either $\class{u}\proclass{v}^{\omega} \cap \lang{\A} = \emptyset$ or $\class{u}\proclass{v}^{\omega} \subseteq \lang{\A}$, it follows that $\infwords \setminus\lang{\A} = \bigcup \setcond{\class{u}\proclass{v}^{\omega}}{u \in \finwords, v \in \poswords, uv\canoEq u, \class{u}\proclass{v}^{\omega} \cap \lang{\A} = \emptyset}$ (the proof is similar to the ones for Lemma~\ref{lem:saturation-impoved-fdfws} and~\ref{lem:saturation-opt-fdfws}).
Thus, $(\canoEq, \bigcup_{u \in \finwords} \setnocond{\proEq_{u}})$ also recognize $\infwords \setminus \lang{\A}$ for all input NBWs $\A$.
According to Definition~\ref{def:improved-ramsey-fdfws} and Definition~\ref{def:induced-dfw}, we are able to construct an FDFW $\F$ accepting $\infwords\setminus \lang{\A}$ with less than $2^{n \log n}$ macrostates for all input NBWs $\A$, contradicting Theorem~\ref{thm:lower-bound-fdfws}.
Thus the claim follows.
\end{proof}

\section{Connection to Preorder of Vertices defined in~\cite{Fogarty13,DBLP:journals/iandc/FogartyKVW15}}
\label{app:comparisonProfile}

In this section, we first introduce the definition of preorder given in~\cite{Fogarty13,DBLP:journals/iandc/FogartyKVW15} and show its connections with the preorder $\preceq_{u}$ we introduced in Sect.~\ref{sec:optimal-rc}.
The preorder in~\cite{Fogarty13,DBLP:journals/iandc/FogartyKVW15} is defined on the vertices of a run directed acyclic graph (run DAG for short), so we first introduce the notion of run DAGs.

\subsection{Run DAGs}
The run directed acyclic graphs were proposed by Kupferman and Vardi in~\cite{kupferman2001weak} for simultaneously reasoning about all runs of an NBW on a given $\omega$-word.
Let $\A = (\states, \inits, \trans, \acc)$ be an NBW and $w = a_1 a_2 \cdots$ be an infinite word.
The run DAG $G_{w, \A} = \langle V, E \rangle$ of $\A$ over $w$ is defined as follows:
\begin{itemize}
\item
    Vertices: $V \subseteq \states \times \naturals$ is the set of vertices  $V = \setcond{\vertex{q}{l} \in \states \times \naturals}{v \in V_{l}, l \geq 1}$ where $V_{1} = \inits$ and $V_{l + 1} = \trans(V_{l}, a_{l})$ for every $l \geq 1$.
\item
    Edges: There is an edge from $\langle q, l\rangle$ to $\langle q', l'\rangle$ if $l' = l + 1$ and $q' \in \trans(q, a_l)$.
\end{itemize}

A vertex $\vertex{q}{l}$ is said to be on level $l$ and there are at most $\size{\states}$ states on each level. A vertex $\vertex{q}{l}$ is an \emph{$\acc$-vertex} if $q \in \acc$.
A finite/infinite sequence of vertices $\hat{\run} = \vertex{q_{1}}{1}\vertex{q_{2}}{2} \cdots$ is called a \emph{branch} of $G_{w, \A}$ if $q_{1} \in \inits$ and for each $1 \leq l < k$, there is an edge from $\vertex{q_{l}}{l}$ to $\vertex{q_{l + 1}}{l+ 1}$;
when the branch is infinite, we set $k = \infty$.
An \emph{$\omega$-branch} of $G_{w, \A}$ is a branch of infinite length, i.e., $k = \infty$.
A finite \emph{fragment} $\vertex{q_{l}}{l} \vertex{q_{l+1}}{l+1} \cdots$ of $\hat{\run}$ is said to be a branch from the vertex $\vertex{q_{l}}{l}$;
a fragment $\vertex{q_{l}}{l} \cdots \vertex{q_{l+k}}{l+k}$ of $\hat{\run}$ is said to be a \emph{path} from $\vertex{q_{l}}{l}$ to $\vertex{q_{l+k}}{l+k}$, where $k \geq 1$.
A vertex $\vertex{q_{j}}{j}$ is \emph{reachable} from $\vertex{q_{l}}{l}$ if there is a path from $\vertex{q_{l}}{l}$ to $\vertex{q_{j}}{j}$.
We call a vertex $\vertex{q}{l}$ \emph{finite} in $G_{w, \A}$ if there are no $\omega$-branches in $G_{w, \A}$ starting from $\vertex{q}{l}$;
we call a vertex $\vertex{q}{l}$ \emph{$\acc$-free} if it is not finite and no $\acc$-vertices are reachable from $\vertex{q}{l}$ in $G_{w, \A}$.

There is a bijection between the set of runs of $\A$ on $w$ and the set of $\omega$-branches in $G_{w,\A}$.
To a run $\run = q_{1} q_{2} \cdots$ of $\A$ over $w$ corresponds an $\omega$-branch $\hat{\run} = \langle q_{1}, 1 \rangle\langle q_{2}, 2 \rangle \cdots$.
Therefore, $w$ is accepted by $\A$ if and only if there exists an $\omega$-branch in $G_{w, \A}$ that visits $\acc$-vertices infinitely often;
we say that such an $\omega$-branch is \emph{accepting}; $G_{w, \A}$ is accepting if and only if there exists an accepting $\omega$-branch in $G_{w, \A}$.

\subsection{Profile of Run DAGs }
\label{app:profile}
The following content is taken from Section 3.1 of~\cite{DBLP:journals/iandc/FogartyKVW15} and adapted to our notation.
Consider the run DAG $G_{w, \A} = \langle V, E\rangle$.
Let $f \colon V \to \setnocond{0, 1}$ be such that $f(\langle q, i\rangle) = 1$ iff $q \in \acc$ and $f(\langle q, i \rangle) = 0$ otherwise.
So $f$ labels $\acc$-vertices by $1$ and all other vertices by $0$.
The \emph{profile} of a branch in $G_{w, \A}$ is the sequence of labels of vertices in that branch.
The profile of a vertex is the lexicographically maximal profile of all branches from initial vertices to that vertex, where an initial vertex is a vertex $\vertex{q}{1}$ such that $q \in \inits$.
The profile of a path  $b = v_{1}, \cdots, v_{k}$ in $G_{w,\A}$ is formalized as $h_{b} = f(v_{1})\cdots f(v_{k})$ and the profile of a branch $b = v_{1}, v_{2}, \cdots$ is $h_{b} = f(v_{1}) f(v_{2}) \cdots$.
Here we define that $h_{b} > h_{b'}$ or $h_{b'} < h_{b}$ if there is a prefix $\alpha 1$ of $h_{b}$ such that $\alpha 0$ is a prefix of $h_{b'}$ and we say that $h_{b}$ is lexicographically \emph{greater} than $h_{b'}$.
Naturally, we have $h_{b} = h_{b'}$ if they are exactly the same binary sequence.
Lastly, the profile of a vertex $v$, written as $h_{v}$, is the lexicographically maximal element of $\setcond{h_{b}}{\text{$b$ is an branch to $v$}}$.
The lexicographic order of profiles of each vertex at the same level induces a linear preorder over vertices on every level of $G_{w,\A}$.
Below we define the linear preorders $\leq_{i}$ on the vertices at level $i \geq 1$ of $G_{w,\A}$.

For vertices $v_{1}$ and $v_{2}$ at level $i$, let $v_{1} \leq_{i} v_{2}$ if $h_{v_{1}} \not > h_{v_{2}}$;
in particular, we have that $v_{1} <_{i} v_{2}$ if $h_{v_{1}} < h_{v_{2}}$ and $v_{1} \simeq_{i} v_{2}$ if $h_{v_{1}} = h_{v_{2}}$.
Similarly to Section~\ref{sec:optimal-rc}, we can group vertices that are equivalent under the preorder $\leq_{i}$.
In particular, the last element of a vertex's profile is $1$ iff that vertex is an $\acc$-vertex.
So vertices in an equivalence class are either all $\acc$-vertices or all not.
In~\cite{DBLP:journals/iandc/FogartyKVW15}, the authors call an equivalence class $\acc$-class when all vertices are $\acc$-vertices and non-$\acc$-class otherwise.
We remark that a state can be reached from initial states by multiple finite runs over the same word due to the nondeterminism of NBWs.
Therefore, a vertex in a run DAG can have multiple incoming edges in $G_{w,\A}$.
So the authors propose to remove from $G_{w,\A}$ all edges that do not contribute to profiles of that vertex.
Formally, we can define the pruned run DAG $G'_{w,\A} = \langle V, E' \rangle$ where $E' = \setcond{( \vertex{p}{i}, \vertex{q}{i+1}) \in E}{\text{for every $\vertex{p'}{i}\in V$, if $(\vertex{p'}{i}, \vertex{q}{i+1}) \in E$, then $\vertex{p'}{i} \leq_{i} \vertex{p}{i}$}}$.
Intuitively, the pruning only keeps the predecessors that are maximal under $\leq_{i}$.
One can see that $G_{w,\A}$ and $G'_{w,\A}$ have the same set of vertices and an edge can be removed from $E$ only when there is another edge to the same destination.

Moreover, the removal of edges does not modify the profiles of vertices and vertices now derive their profiles from their parents in $G'_{w,\A}$, as formalized below.
\begin{lemma}[Lemma~3.1 of \cite{DBLP:journals/iandc/FogartyKVW15}]
For two vertices $u$ and $v$ in $V$, if $\langle u, v\rangle \in E'$, then $h_{v} = h_{u}0$ or $h_{v} = h_{u} 1$.
\end{lemma}

The vertices with different profiles can share a child in $G_{w,\A}$;
the following lemma precludes this from happening in $G'_{w,\A}$.
\begin{lemma}[Lemma~3.2 of \cite{DBLP:journals/iandc/FogartyKVW15}]
\label{lem:mapOfSubclass}
Let vertices $u$ and $v$ be at level $i$ in $G'_{w,\A}$ and vertices $u'$ and $v'$ at level $i+1$ in $G'_{w,\A}$.
If $\langle u, u' \rangle \in E'$, $\langle v, v' \rangle \in E'$ and $u' \simeq_{i+1} v'$, then $u \simeq_{i} v$.
\end{lemma}

Despite the removal of some edges in $G_{w,\A}$, $G'_{w,\A}$ is accepting iff $G_{w,\A}$ is accepting, as stated by following result.
\begin{lemma}[Theorem~3.3 of \cite{DBLP:journals/iandc/FogartyKVW15}]
The pruned run DAG $G'_{w,\A}$ is accepting iff $\A$ accepts $w$.
\end{lemma}

\subsection{Connection to Our Preorder $\preceq_{u}$}
First, we point out that the function $f$ in Sect.~\ref{app:profile} for defining profile is a counterpart to our function $\dirac$ in Section~\ref{sec:optimal-rc}.
The difference is that $f$ works on a branch in the run DAG $G_{w,\A}$ while $\dirac$ operates on the runs of $\A$.
Inherently, $f$ can be seen as the extension of $\dirac$ from runs of $\A$ to branches in the run DAG $G_{w,\A}$.

We would like to point out that the profile of a vertex $\vertex{q}{i}$ corresponds to the set of maximal runs $\max(\Pi_{w[1\cdots i-1]}, q)$ where $i \geq 0$.
Recall that if $i = 1$, $w[1\cdots 0] = \emptyword$.
This is because the profile of $\vertex{q}{i}$ only keeps the lexicographically maximal profiles from the initial vertices to $\vertex{q}{i}$;
lexicographically maximal profiles correspond to our maximal runs in $\max(\Pi_{w[1\cdots i-1]}, q)$.

The profiles of vertices at the same level $i$ induce a linear preorder $\leq_{i}$ which corresponds to the preorder $\preceq_{w[1\cdots i-1]}$.
One can see that $\vertex{q}{i} \leq_{i} \vertex{q'}{i}$ if and only if we have that $q \preceq_{w[1\cdots i-1]} q'$ and $\vertex{q}{i} <_{i} \vertex{q'}{i}$ if and only if we have that $q \prec_{w[1\cdots i-1]} q'$.

When pruning $G_{w,\A}$ to the run DAG $G'_{w,\A}$, from the definition of $E'$, we can see that the pruning only keeps the maximal predecessors at the previous level of a vertex.
Therefore we can observe that Lemma~\ref{lem:mapOfSubclass} is the counterpart to our Lemma~\ref{lem:subseteq-class-u}.
In our Lemma~\ref{lem:subseteq-class-u}, we proved that along the maximal runs to two different states in the same equivalence class, the states are stepwise equivalent under the preorder $\preceq_{u}$ we defined.
This is also reflected in Lemma~\ref{lem:mapOfSubclass} that along the branch in the pruned $G'_{w,\A}$ to equivalent vertices $u'$ and $v'$ at level $i+1$, their respective predecessors $u$ and $v$ are also equivalent.
Finally, the pruning does not change the acceptance of $G'_{w,\A}$ with respect to $G_{w,\A}$.
This also holds for our presentation of maximal runs as given in Lemma~\ref{lem:existence-max-run} (in Appendix~\ref{app:opt-rc}).

\section{Comparison with~\cite{calbrix1993ultimately} and~\cite{Kuperberg19dlt}}
\label{app:comparison}

First, we say a decomposition $(u, v)$ is \emph{captured} by $\F$ if we have $v \in \finlang{\N_{q}}$  where $q = \M(u)$.
So a decomposition $(u, v)$ is accepted by $\F$ if $(u, v)$ is normalized and captured by $\F$.

The FDFW $\F$ obtained from~\cite{calbrix1993ultimately} for $\A$ is represented with a DFW $\D_{\$}$ such that $\finlang{\D_{\$}} = \setcond{u \$ v \in \finwords \setnocond{\$} \poswords}{uv^{\omega} = w, w \in \lang{\A}}$ where $\$ \notin \alphabet$.
This construction has been further improved in~\cite{Kuperberg19dlt}, which produces a DFW $\D_{\$}$ with $\bigO(3^{n^{2} + n})$ states.
Both these two constructions yield an FDFW $\F$ such that each decomposition $(u, v)$ of $w \in \lang{\A}$ is captured by $\F$, as $\D_{\$}$ accepts each $u\$v$ such that $uv^{\omega} \in \lang{\A}$.
In contrast, ours only captures normalized decompositions of $w$, see the condition $uv \canoEq^{o} u$ in Definition~\ref{def:improved-ramsey-fdfws}, which can be changed to $uv \icanoEq u$ when dealing with $\icanoEq$ and $\proEq_{u}$.

Our direct construction for FDFWs with Definition~\ref{def:improved-ramsey-fdfws} replaced with $\icanoEq$ and $\proEq_{u}$ has the same worst-case complexity as the one in~\cite{Kuperberg19dlt} since $\size{\proEq_{u}} \leq 3^{n^{2}}$.
Nonetheless, we show that the FDFWs constructed in~\cite{calbrix1993ultimately, Kuperberg19dlt} may be exponentially larger than the FDFWs by our construction even for DBWs, since the FDFWs constructed in~\cite{calbrix1993ultimately, Kuperberg19dlt} try to capture each decomposition of desired UP-words.

\begin{restatable}{theorem}{comparisonCNP}
\label{thm:comparison-with-cnp}
There exists a family of DBWs $\C_{1}, \cdots, \C_{n}$ with $n+2$ states for which the FDFW obtained from~\cite{calbrix1993ultimately, Kuperberg19dlt} has $2^{\Omega(n \log n)}$ macrostates while the FDFW constructed in Definition~\ref{def:improved-ramsey-fdfws} has $\bigO(n^{2})$ macrostates.
\end{restatable}
\begin{proof}
Each DBW $\C_{n}$ is obtained from the NBW $\B_{n}$ depicted in Fig.~\ref{fig:family-nbw-ryc-smaller} by collapsing $q_{-1}$ and $q_{0}$ together as a single nonaccepting sink state;
the resulting DBWs $\C_{n}$ are the same automata considered in~\cite[Theorem 2]{AngluinF16}.
It is immediate to see that $\C_{n}$ is a DBW.
The right congruences used in Definition~\ref{def:improved-ramsey-fdfws} are $(\canoEq^{o}, \cup_{u\in\finwords} \setnocond{\proEq^{o}_{u}})$; according to Lemma~\ref{lem:size-dbw-opt-fdfw}, we have that $\Sigma_{\oclass{u} \in \finwords/_{\canoEq^{o}}} \size{\proEq^{o}_{u}} \in \bigO(n^{2})$, i.e., the number of equivalence classes defined by all right congruences $\cup_{\finwords/_{\canoEq^{o}}} \size{\proEq^{o}_{u}}$ is in $\bigO(n^2)$;
moreover, in the proof of Lemma~\ref{lem:size-dbw-opt-fdfw}, we have showed that $\size{\canoEq^{o}} \leq n$.
It follows that the number of equivalence classes defined with $(\canoEq^{o}, \cup_{u\in\finwords} \setnocond{\proEq^{o}_{u}})$ is also in $n + \bigO(n^2) \in \bigO(n^2)$.
Since Definition~\ref{def:improved-ramsey-fdfws} gives a linear translation from $(\canoEq^{o}, \cup_{u\in\finwords} \setnocond{\proEq^{o}_{u}})$ to an FDFW $\F$, it follows that the number of macrostates in $\F$ is also in $\bigO(n^2)$.

We remark that when we use $(\icanoEq, \cup_{u\in\finwords} \setnocond{\proEq_{u}})$, we can similarly prove that the constructed FDFW $\F$ has $\bigO(n^2)$ macrostates according to the proof of Theorem~\ref{thm:size-dbw-improved-fdfw}.

According to~\cite[Theorem 2]{AngluinF16}, the minimal FDFW $\F'$ capturing each decomposition $(u, v)$ of $w \in \lang{\C_{n}}$ has at least $n! \in 2^{\bigO(n \log n)}$ macrostates regardless of the transition structure of $\C_{n}$.
Since the FDFWs obtained from ~\cite{calbrix1993ultimately, Kuperberg19dlt} capture exactly the set of decompositions $(u, v)$ of $w \in \lang{\C_{n}}$, their constructed FDFWs from $\C_{n}$ must have $2^{\Omega(n \log n)}$ macrostates.
\end{proof}

\end{document}